\documentclass[journal,10pt,twocolumn,twoside]{IEEEtran}
\IEEEoverridecommandlockouts
\usepackage{amsmath}
\usepackage{amsfonts}
\usepackage{cases}
\usepackage{setspace}
\usepackage{fancybox}
\usepackage{subfigure}
\usepackage{epsfig}
\usepackage{graphicx}
\usepackage{epstopdf}
\usepackage{float}
\usepackage{multirow}
\usepackage{color}
\usepackage{amsmath}
\usepackage{multirow}
\usepackage{indentfirst}
\usepackage{dsfont}
\usepackage{amsfonts}
\usepackage{times,amsmath,color,amssymb,epsfig,cite,subfigure,algorithm,algorithmic}

\newenvironment{proof}[1][Proof]{\begin{trivlist}
\item[\hskip \labelsep {\bfseries #1}]}{\end{trivlist}}

\newtheorem{lemma}{Lemma}
\newtheorem{theorem}{Theorem}

\usepackage{color}

\begin{document}
\title{Efficient QAM Signal Detector for Massive MIMO Systems via PS-ADMM Approach}
\author{Quan~Zhang, Jiangtao Wang, Yongchao~Wang,~\IEEEmembership{Senior~Member,~IEEE}
}

\markboth{}
{}

\maketitle

\begin{abstract}
In this paper, we design an efficient quadrature amplitude modulation (QAM) signal detector for massive multiple-input multiple-output (MIMO) communication systems via the penalty-sharing alternating direction method of multipliers (PS-ADMM). Its main content is as follows:
first, we formulate QAM-MIMO detection as a maximum-likelihood optimization problem with bound relaxation constraints.
Decomposing QAM signals into a sum of multiple binary variables and exploiting introduced binary variables as penalty functions, we transform the detection optimization model to a non-convex sharing problem;
second, a customized ADMM algorithm is presented to solve the formulated non-convex optimization problem. In the implementation, all variables can be solved analytically and in parallel;
third, it is proved that the proposed PS-ADMM algorithm converges under mild conditions.
Simulation results demonstrate the effectiveness of the proposed approach.
\end{abstract}


\IEEEpeerreviewmaketitle

\section{Introduction}
\IEEEPARstart{M}{assive} multiple-input multiple-output (MIMO) technology is considered to be one of the disruptive technologies for the fifth-generation (5G) communication systems \cite{marzetta2010noncooperative}, \cite{Boccardi2013Five}. The foremost benefit of massive MIMO is significant increase in the spatial degrees of freedom which can improve throughput and energy efficiency significantly in comparison with conventional MIMO systems \cite{Ngo2013Energy}.
However, numerous practical challenges arise in implementing massive MIMO technology in order to achieve such improvements.
One such challenge is signal detection lying in the uplink massive systems since it is difficult to achieve an effective compromise among good detecting performance, low computational complexity, and high processing  parallelism \cite{albreem2019massive}.

The optimal MIMO detection is the maximum-likelihood (ML) detector \cite{Verdu1986Minimum}, which suffers from an exponential increase in computational complexity with an increasing number of terminal antennas and entails prohibitive complexity for massive MIMO detection \cite{albreem2019massive}. Various nonlinear detection methods, such as sphere decoding (SD) \cite{damen2003maximum}, semidefinite relaxation (SDR) \cite{Zhi2010Semidefinite}, PSK detector \cite{luo2003efficient} and K-best \cite{guo2006algorithm}, have achieved near-optimal ML detection performance in small-scale MIMO systems \cite{yang2015fifty}, but they are still prohibitively complex for large-scale or high-order modulation MIMO systems.
Linear detection methods such as minimum mean square error (MMSE) \cite{6771364} and zero-forcing (ZF) \cite{1092893}, are one of the best choices with respect to (w.r.t.) the tradeoff in performance and complexity, especially when the base station (BS)-to-user-antenna ratio is large \cite{rusek2012scaling}. In order to further reduce computational complexity and improve performance,
on the one hand, numerous detectors that can be called as approximate matrix inversion methods have been proposed to reduce the computational complexity of linear detectors \cite{albreem2019massive}, such as Neumann series (NS) \cite{wu2014large}, \cite{vcirkic2014complexity}, Gauss-Seidel (GS) \cite{dai2015low}, \cite{wu2016efficient}, Richardson (RI) \cite{gao2014low}, and conjugate gradient (CG) \cite{yin2014conjugate} methods. However, the decrease in computational complexity of these algorithms comes at the expense of the loss of detection performance; meanwhile, these methods also deliver poor bit error rate (BER) performance when the BS-to-user antenna ratio is close to one.
On the other hand, various nonlinear detection algorithms have been improved to obtain better detection performance than linear detection algorithms for massive MIMO systems. In \cite{wu2016high}, an optimized coordinate descent with the box-constrained equalization (OCD-BOX) detector shows better BER performance with low hardware complexity; however, it can't be implemented in parallel. In \cite{castaneda2016data}, the triangular approximate semidefinite relaxation (TASER) detector achieves near-ML performance while providing comparable hardware-efficiency, but extending TASER to support higher-order modulation schemes is the subject of ongoing research.

In recent years, the alternating direction method of multipliers (ADMM) technique is widely used to solve convex and non-convex problems due to its simplicity, operator splitting capability, and guaranteed-convergence under mild conditions \cite{boyd2011distributed}. The ADMM strategy was first introduced in MIMO detection by Takapoui et al. in \cite{7526551}. In \cite{shahabuddin2021admm}, a detection algorithm based on ADMM with infinity norm or box-constrained equalization, named ADMIN, was proposed, which outperformed the state-of-the-art linear detectors by a large margin if the BS-to-user-antenna ratio is small.
The ADMM detection algorithm was applied in various scenarios for MIMO systems \cite{souto2016mimo}, \cite{souto2018efficient} and improved for massive MIMO systems \cite{lopes2018iterative,elgabli2019a,elgabli2019pro,feng2019a,8755566}. Although the existing ADMM-based methods provide better BER performance than conventional detectors, there are two major drawbacks: first, the constraints set of optimization problem is over-relaxed and second, these works cannot be applied to high-order modulation systems directly.

In this paper, we focus on designing a new ADMM-based quadrature amplitude modulation (QAM) signal detector for massive MIMO systems. By exploiting ideas of penalized bound relaxation for the ML detection formulation \cite{ma2002quasi}, binary transformation for high-order QAM signals \cite{mao2007semidefinite}, and sharing-ADMM technique \cite{boyd2011distributed}, \cite{hong2016convergence}, we obtain a new detector for massive MIMO systems called PS-ADMM, which has favorable BER performance and cheap computational cost. The main contributions of this paper are summarized as follows:

\begin{itemize}
\item \emph{Penalty sharing formulation}:
  the well known MIMO detection ML formulation is equivalent to a non-convex sharing problem. In it, high-order modulation signals are decomposed into a sum of multiple binary ones and penalty functions involving these binaries are introduced in the sharing formulation.
\item \emph{Efficient implementation}: sharing-ADMM algorithms are designed to solve the formulated optimization model. In the implementation, all the variables in subproblems can be solved analytically, accurately, and updated in parallel in each iteration step. As a result, the PS-ADMM algorithm achieves favourable BER performance with cheap computational complexity, especially when the BS-to-user-antenna ratio is close to one.
\item \emph{Theoretically-guaranteed performance}:
we prove that the proposed PS-ADMM algorithm is convergent arbitrarily close to a stationary point of the non-convex optimization problem if proper parameters are chosen.
\end{itemize}

The rest of this paper is organized as follows.
In Section \ref{sec:problem-formulation}, we formulate the massive MIMO detection problem to a non-convex sharing problem.
In Section \ref{sec:PS-ADMM Solving Algorithm}, an efficient sharing-ADMM algorithm is customized to solve the formulated problem.
Section \ref{sec:Analysis-psadmm} presents the detailed performance analysis, including convergence, iteration complexity, and computational cost of the proposed PS-ADMM algorithm.
Simulation results, which show the effectiveness of our proposed PS-ADMM algorithm, are presented in Section \ref{sec:Simulation results}
and the conclusions are given in Section \ref{sec:Conclusion}.

\emph{Notations}: In this paper, bold lowercase, uppercase, and italic letters denote column vectors, matrices, and scalars respectively;
 $\mathds{C}$ denotes the complex field;
 $(\cdot)^H$ symbolizes the conjugate transpose operation;
 $x_{iR}$, $x_{iR}$ denote the real and imaginary parts of the $ith$ entry of vector $\mathbf{x}$ respectively;
 $\|\cdot\|_2$ represents the 2-norm of vector $\mathbf{x}$;
 $\underset{[a,b]}\Pi(\cdot)$ denotes the Euclidean projection operator onto interval $[a,b]$;
 $\nabla(\cdot)$ represents the gradient of a function;
 $\rm{Re}(\cdot)$ takes the real part of the complex variable;
 $\lambda_{\rm min}(\cdot)$ and $\lambda_{\rm max}(\cdot)$ denote the minimum and maximum eigenvalues of a matrix respectively;
 $\langle \cdot,\cdot\rangle$ denotes the dot product operator; and ${\mathbf{I}}$ denotes an identity matrix.
\vspace{-0.1cm}
\section{System Model And Problem Formulation}
\label{sec:problem-formulation}

The considered signal detection problem lies in uplink multiuser massive MIMO systems, where BS equipped with $B$ antennas serves $U$ single-antenna users. Here, we assume $ B \ge U$. Typically, the received signal vector at BS can be characterized by the following model
\begin{equation}\label{transmission model}
{\mathbf{r}} = {\mathbf{Hx}} + {\mathbf{n}},
\end{equation}
where $\mathbf{x}\in \mathcal{X}^{U}$ is the transmitted signal vector from $U$ users and $\mathcal{X}$ refers to the signal constellation set, $\mathbf{r}\in \mathbb{C}^{B}$ is the BS received signal vector, $\mathbf{H}\in \mathbb{C}^{B\times U}$ denotes the MIMO channel matrix, and $\mathbf{n}\in \mathbb{C}^{B}$ denotes additive white Gaussian noise. The entries of $\mathbf{H}$ and $\mathbf{n}$ are assumed to be independent and identically distributed (i.i.d.) complex Gaussian variables with zero mean.

The optimal MIMO ML detector for $4^Q$-QAM signals, i.e., achieving minimum error probability of detecting $\mathbf{x}$ from the received signal $\mathbf{r}$ can be formulated as the following discrete least square problem \cite{Zhi2010Semidefinite}
{\setlength\abovedisplayskip{2pt}
\setlength\belowdisplayskip{2pt}
\setlength\jot{2pt}
\begin{equation}\label{eq:MLdetection}
\begin{split}
&\min_{\mathbf{x}\in \mathcal{X}^U} \Vert\mathbf{r}-\mathbf{H}\mathbf{x} \Vert_2^{2}, \\
\end{split}
\end{equation}
where} ${\mathcal X}= \{x= x_{\rm R} + j x_{\rm I} \vert x_{\rm R}, x_{\rm I} \in \{\pm 1, \pm 3, \cdots , \pm (2^{Q} -1)\}\}$ and $Q$ is some positive integer. Obtaining its global optimal solution is prohibitive in practice since the corresponding computational complexity grows exponentially with the users' number $U$, BS's antenna number $B$, and set $\mathcal{X}$'s size \cite{chen2017a}. In the following, by exploiting insight structures of the model \eqref{eq:MLdetection}, we propose a {\it relaxation-tighten} technique and transform it to the well-known sharing problem.

Any transmitted signal vector $\mathbf{x}\in\mathcal{X}^{U}$ can be expressed by
{\setlength\abovedisplayskip{2pt}
\setlength\belowdisplayskip{2pt}
\begin{equation}\label{eq:XsumXq}
\mathbf{x}=\sum_{q=1}^{Q}2^{q-1} \mathbf{x}_q,\ q = 1,\cdots,Q.
\end{equation}
Plugging} \eqref{eq:XsumXq} into the model \eqref{eq:MLdetection}, it can be equivalent to
{\setlength\abovedisplayskip{2pt}
 \setlength\belowdisplayskip{2pt}
 \setlength\jot{2pt}
\begin{subequations}\label{eq:S_MLdetection}
\begin{align}
&\hspace{0.4cm}\min_{\mathbf{x}_q} \hspace{0.2cm} \frac{1}{2}\Vert\mathbf{r}-\mathbf{H}(\sum_{q=1}^{Q}2^{q-1} \mathbf{x}_q) \Vert_2 ^{2},  \label{eq:S_MLdetection_a} \\
&\hspace{0.5cm}{\rm {s.t.}}   \hspace{0.3cm}\mathbf{x}_q \in \mathcal{X}_q^{U},\ q = 1,\cdots,Q, \label{eq:S_MLdetection_b}
\end{align}
\end{subequations}
where} $\mathcal{X}_q=\{\mathbf{x}_q = \mathbf{x}_{q \rm R} + j \mathbf{x}_{q \rm I}| \mathbf{x}_{q\rm R}, \mathbf{x}_{q\rm I} \in \{1,-1\}\}$. Relaxing the binary constraints in \eqref{eq:S_MLdetection_b} to the box constraints and then tightening the relaxation by adding the quadratic penalty function into the objective \eqref{eq:S_MLdetection_a}, we formulate the following model
{\setlength\abovedisplayskip{2pt}
  \setlength\belowdisplayskip{2pt}
   \setlength\jot{2pt}
\begin{equation}\label{eq:PS_ML}
\begin{split}
&\hspace{0.3cm} \min_{\mathbf{x}_q} \hspace{0.2cm} \frac{1}{2}\Vert\mathbf{r}-\mathbf{H}(\sum_{q=1}^{Q}2^{q-1} \mathbf{x}_q) \Vert_2 ^{2}-\sum_{q=1}^{Q}\frac{\alpha_q}{2}\Vert\mathbf{x}_q \Vert_2 ^{2},\\
&\hspace{0.4cm} {\rm {s.t.}} \hspace{0.4cm}\mathbf{x}_q\in \tilde{\mathcal{X}}_q^{U},\;q = 1,\cdots,Q,
\end{split}
\end{equation}
where} penalty parameters $\alpha_q \ge 0$ and $\tilde{\mathcal{X}}_q= \{\mathbf{x}_{q\rm R} + j \mathbf{x}_{q\rm I} \vert\mathbf{x}_{q \rm R}, \mathbf{x}_{q \rm I} \in [-1\ 1]\}$. The model \eqref{eq:PS_ML} can be cast as the sharing problem \cite[Section 7.3]{boyd2011distributed}, but with the important difference that the objective function is non-convex. It is easy to see that the non-convex quadratic penalty function can make the integer solutions more favorable.

\section{PS-ADMM Solving Algorithm}\label{sec:PS-ADMM Solving Algorithm}

In this section, we develop an efficient ADMM algorithm, named PS-ADMM, to solve the model \eqref{eq:PS_ML}. First, to facilitate distributed computation, we transform it equivalently to a linearly constrained problem by introducing auxiliary variable $\mathbf{x}_0\in \mathbb{C}^{U}$ 
{\setlength\abovedisplayskip{2pt}
    \setlength\belowdisplayskip{2pt}
      \setlength\jot{2pt}
\begin{equation}\label{eq:PS_ADMM}
\begin{split}
&\min_{\mathbf{x}_0, \mathbf{x}_q} \hspace{0.2cm} \frac{1}{2}\Vert\mathbf{r}-\mathbf{H}\mathbf{x}_0 \Vert_2 ^{2}-\sum_{q=1}^{Q}\frac{\alpha_q}{2}\Vert\mathbf{x}_q \Vert_2 ^{2},\\
&\ {\rm {s.t.}} \hspace{0.3cm} \mathbf{x}_0=\sum_{q=1}^{Q}2^{q-1} \mathbf{x}_q,\ \mathbf{x}_q\in  \tilde{\mathcal{X}}_q^{U}, q = 1,\cdots,Q.
\end{split}
\end{equation}
The} augmented Lagrangian function of the model \eqref{eq:PS_ADMM} can be expressed as
{\setlength\abovedisplayskip{2pt}
    \setlength\belowdisplayskip{2pt}
      \setlength\jot{2pt}
\begin{align}
\begin{split}\label{eq:lagrangian_PSADMM}
&L_{\rho}(\{\mathbf{x}_q\}_{q=1}^{Q},\mathbf{x}_0,\mathbf{y}) = \frac{1}{2}\Vert\mathbf{r}\!-\!\mathbf{H}\mathbf{x}_0 \Vert_2 ^{2}\!-\!\sum_{q=1}^{Q}\frac{\alpha_q}{2}\!\Vert\mathbf{x}_q\!\Vert_2 ^{2}\\
&\ \ +{\rm Re}\big\langle \mathbf{x}_0-\sum_{q=1}^{Q}2^{q-1} \mathbf{x}_q, \mathbf{y}\big\rangle+\frac{\rho}{2}\big\|\mathbf{x}_0-\sum_{q=1}^{Q}2^{q-1}\mathbf{x}_q \big\|_2^2,
\end{split}
\end{align}
where} $\mathbf{y}\in \mathbb{C}^{U}$ and $\rho>0$ are the Lagrangian multiplier and penalty parameter respectively.
Based on the above augmented Lagrangian, the proposed PS-ADMM solving algorithm framework can be described as
\begin{subequations}\label{PSADMM_update_ori}
\begin{align}
&\mathbf{x}_q^{k+1}\!\! =\!  \mathop{\arg\min}_{\mathbf{x}_q\in\tilde{\mathcal{X}}_q^{U}} L_{\rho}(\!\mathbf{x}_1^{k+1}\!\!\!\!,\!\cdots\!,\mathbf{x}_{q-1}^{k+1}\!,\mathbf{x}_q, \mathbf{x}_{q+1}^k\!,\!\cdots\!,\!\mathbf{x}_Q^k,\!\mathbf{x}_0^{k},\! \mathbf{y}^{k}),\nonumber\\
& \hspace{5.8cm} q=1,\cdots,Q,  \label{eq:x_q_update}\\
&\mathbf{x}_0^{k+1}=\mathop{\arg\min}_{\mathbf{x}_0} \;\;L_{\rho}( \{\mathbf{x}_q^{k+1}\}_{q=1}^{Q}, \mathbf{x}_0, \mathbf{y}^{k}), \label{eq:x0_update}\\
&\mathbf{y}^{k+1}=\mathbf{y}^k+\rho\Big(\mathbf{x}_0^{k+1}-\sum_{q=1}^{Q}2^{q-1} \mathbf{x}_q^{k+1}\Big),\label{eq:y_update}
\end{align}
\end{subequations}
where $k$ denotes iteration number.

The main challenge of implementing \eqref{PSADMM_update_ori} lies in how to solve optimization subproblems \eqref{eq:x_q_update} and \eqref{eq:x0_update} efficiently.
For \eqref{eq:x_q_update}, it can be observed that $L_{\rho}(\{\mathbf{x}_q\}_{q=1}^{Q},
\mathbf{x}_0^k, \mathbf{y}^k)$ is a strongly convex quadratic function w.r.t. $\mathbf{x}_q$ when $4^{q-1}\rho-\alpha_q>0$. It means that, if we set $4^{q-2}\rho>\alpha_q$,  the optimal solution of the subproblems \eqref{eq:x_q_update} can be obtained through the following procedure:

Set the gradient of the corresponding augmented Lagrangian function w.r.t. $\mathbf{x}_q$ to be zero, i.e.,
\begin{equation}
\nabla_{\mathbf{x}_q}L_{\rho}(\mathbf{x}_q,\!\mathbf{x}_1^{k+1}\!\!,\!\cdots\!,\mathbf{x}_{q-1}^{k+1}\!, \mathbf{x}_{q+1}^k,\cdots,\mathbf{x}_Q^k,\mathbf{x}_0^{k},\!\mathbf{y}^{k}) \!= \!0,
\end{equation}
which leads to the following linear equation
\begin{equation}\label{eq:x_q_update_zero}
\begin{split}
& \nabla_{\mathbf{x}_q} \Big(-\frac{\alpha_q}{2}\Vert\mathbf{x}_q \Vert_2 ^{2}-{\rm Re}\langle 2^{q-1}\mathbf{x}_q,\;\mathbf{y}^k\rangle  \\
&\!+\!\frac{\rho}{2}\|\mathbf{x}_0^{k}\!-\!\!\sum_{i<q} 2^{i-1}\mathbf{x}_i^{k+1}\!-\!\!\sum_{i>q}2^{i-1} \mathbf{x}_i^{k}\!-\!2^{q-1} \mathbf{x}_q \|_2^2 \Big)\!=\!0.
\end{split}
\end{equation}
By solving \eqref{eq:x_q_update_zero}, we can obtain
\begin{equation}\label{eq: solution_Xq}
\begin{split}
&\mathbf{x}_q^{k+1} = \underset{[-1,1]}\Pi \bigg(\frac{2^{q-1}}{4^{q-1}\rho-\alpha_q}\Big( \rho\mathbf{x}_0^{k}-\rho\sum_{i<q} 2^{i-1}\mathbf{x}_i^{k+1} \\
&\hspace{0.9cm} -\rho\sum_{i>q}2^{i-1} \mathbf{x}_i^{k}+ \mathbf{y}^k \Big) \bigg),\ \ \ q = 1,\cdots,Q,
\end{split}
\end{equation}
where operator $\underset{[-1,1]}\Pi(\cdot)$ projects every entry's real part and imaginary part of the input vector onto [-1,1] respectively.

Moreover, $L_{\rho}(\{\mathbf{x}_q^{k+1}\}_{q=1}^{Q}, \mathbf{x}_0, \mathbf{y}^k)$ is a strongly convex quadratic function w.r.t. $\mathbf{x}_0$ since $\rho>0$ and matrix $\mathbf{H}^{H} \mathbf{H} $ is positive semidefinite, hence the optimal solution of the subproblem \eqref{eq:x0_update} can also be obtained by setting $\nabla_{\mathbf{x}_0} L_{\rho}(\{\mathbf{x}_q^{k+1}\}_{q=1}^{Q}, \mathbf{x}_0, \mathbf{y}^{k}) $ to be zero and solving the corresponding linear equation, which results in
\begin{equation}\label{eq: solution X0}
\mathbf{x}_0^{k+1}\!=\!{(\mathbf{H}^{H} \mathbf{H}\!+\!\rho\mathbf{I})}^{-1}\bigg(\! \mathbf{H}^{H} \mathbf{r}\!+\! \rho\sum_{q=1}^{Q}2^{q-1} \mathbf{x}^{k+1}_q\!-\!\mathbf{y}^{k}\!\bigg).
\end{equation}

 To be clear, we summarize the proposed PS-ADMM algorithm for solving model \eqref{eq:PS_ADMM} in \emph{Algorithm \ref{PS-ADMM algorithm}}.

\begin{algorithm}[!t]
\caption{The proposed PS-ADMM algorithm}
\label{PS-ADMM algorithm}
\begin{algorithmic}[1]
\REQUIRE $\mathbf{H}$, $\mathbf{r}$, ${Q}$, $\rho$, $\{\alpha_q\}_{q=1}^{Q}$\\
\ENSURE $\mathbf{x}_0^k$\\
\STATE Initialize $\{\mathbf{x}_q^{1}\}_{q=1}^{Q}, \mathbf{x}_0^{1}, \mathbf{y}^1$ as the all-zeros vectors\footnotemark.
\STATE  \textbf{For $k = 1,2,\cdots$}
\STATE \hspace{0.2cm} Step 1: Update $\{\mathbf{x}_q^{k+1}\}_{q=1}^{Q}$ sequentially via \eqref{eq: solution_Xq}.
\STATE \hspace{0.2cm} Step 2: Update $\mathbf{x}_0^{k+1}$ via \eqref{eq: solution X0}.
\STATE \hspace{0.2cm} Step 3: Update $\mathbf{y}^{k+1}$ via \eqref{eq:y_update}.
\STATE \textbf{Until} some preset condition is satisfied.
\end{algorithmic}
\end{algorithm}
\footnotetext{Unlike some existing algorithms that are derived from the solution of the simple linear algorithm such as MMSE or ZF, there are no computations required for PS-ADMM to perform initialization, and the initial values $\{\mathbf{x}_q^{1}\}_{q=1}^{Q}$, $\mathbf{x}_0^{1}$, $\mathbf{y}^{1}$ can be set to zeros, ones, minus ones, and random values.}
\vspace{-0.1cm}

\section{Performance Analysis}
\label{sec:Analysis-psadmm}

In this section, a detailed analysis of the proposed PS-ADMM algorithm on convergence, iteration complexity, and computational cost are provided.

\subsection{Convergence property}
We have the following theorem to show convergence properties of the proposed PS-ADMM algorithm.
\begin{theorem}\label{thm:convergence}
Assume parameters $\rho$ and $\alpha_q$ satisfy $4^{q-1}\rho>\alpha_q$ and $\rho > \sqrt{2} \lambda_{\rm max}(\mathbf{H}^H\mathbf{H})$, where $q = 1,\cdots,Q$. Then, sequence $\{\{\mathbf{x}^{k}_q\}_{q=1}^{Q}, \mathbf{x}_0^{k}, \mathbf{y}^{k}\}$ generated by \emph{Algorithm \ref{PS-ADMM algorithm}} is convergent, i.e.,
\begin{equation}\label{convergence variables}
\begin{split}
&\lim\limits_{k\rightarrow+\infty}\mathbf{x}^{k}_q=\mathbf{x}^*_q, \ \ \lim\limits_{k\rightarrow+\infty}\mathbf{x}_0^{k}=\mathbf{x}_0^*, \lim\limits_{k\rightarrow+\infty}\mathbf{y}^{k}=\mathbf{y}^*, \\
& \hspace{0.3cm} \forall~\mathbf{x}_q\in \tilde{\mathcal{X}}_q^{U}, \; q = 1,\cdots,Q.
\end{split}
\end{equation}
Moreover, $\{\mathbf{x}^*_q\}_{q=1}^{Q}$ is a stationary point of problem \eqref{eq:PS_ML}, i.e., $\forall~\mathbf{x}_q\in \tilde{\mathcal{X}}_q^{U},\; q = 1,\cdots,Q$, which satisfies
\begin{align}\label{stationary point}
 {\rm Re}\Big \langle\nabla_{\mathbf{x}_q} \Big(\ell\big(\sum_{q=1}^{Q}2^{q-1} \mathbf{x}_q^*\big)-\sum_{q=1}^{Q}\frac{\alpha_q}{2}\Vert\mathbf{x}_q^* \Vert_2 ^{2}\Big),\mathbf{x}_q-\mathbf{x}^*_q\Big\rangle
 \ge 0,
\end{align}
\end{theorem}
where
\begin{equation}\label{loss}
\ell\big(\sum_{q=1}^{Q}2^{q-1} \mathbf{x}_q^*\big)=\frac{1}{2}\Vert\mathbf{r}-\mathbf{H}(\sum_{q=1}^{Q}2^{q-1} \mathbf{x}_q^*) \Vert_2 ^{2}.
\end{equation}

{\it Remarks:}
   Theorem \ref{thm:convergence} indicates that the proposed PS-ADMM algorithm is theoretically-guaranteed converged to some stationary point of the model \eqref{eq:PS_ML} under some wild conditions. Here, we should note that these conditions are easily satisfied since the values of penalty parameters $\rho$ and $\{\alpha_q\}_{q=1}^{Q}$ can be set accordingly when the channel matrix $\mathbf{H}$ is known. The key idea of proving Theorem \ref{thm:convergence} is to find out that the potential function $L_{\rho}(\{\mathbf{x}_q\}_{q=1}^{Q}, \mathbf{x}_0, \mathbf{y})$ {\it decreases sufficiently} in every PS-ADMM iteration and is lower-bounded. To reach this goal, we first prove several related lemmas in Appendix \ref{lemma1-3}. Then, we provide detailed proof of Theorem \ref{thm:convergence} in Appendix \ref{PS-ADMM Proof}.

\subsection{Iteration complexity}

We use the following residual
\begin{equation}\label{residual_k}
\displaystyle\sum_{q=1}^{Q}\|\mathbf{x}_q^{k+1}-\mathbf{x}_{q}^{k}\|_2^2 + \|\mathbf{x}_0^{k+1}-\mathbf{x}_0^{k}\|_2^2
\end{equation}
 to measure the convergence progress of the PS-ADMM algorithm since it converges to zero as $k\rightarrow+\infty$. Then, we have Theorem \ref{iteration complexity} about its convergence progress. The detailed proof is shown in Appendix \ref{Iteration complexity Proof}.

\begin{theorem}\label{iteration complexity}
  Let $t$ be the minimum iteration index such that the residual in \eqref{residual_k} is less than $\epsilon$, where $\epsilon$ is the desired precise parameter for the solution. Then, we have the following iteration complexity result
 \[
   \begin{split}
     t \!\leq \!\frac{1}{C\epsilon}\bigg(\!L_{\rho}(\{\mathbf{x}^{1}_q\}_{q=1}^{Q}, \mathbf{x}_0^{1}, \mathbf{y}^{1}) \!-\! \Big(\ell\left(\mathbf{x}_0^*\right) \!-\! \sum_{q=1}^{Q}\frac{\alpha_q}{2} \Vert\mathbf{x}_q^* \Vert_2 ^{2}\Big)\!\bigg),
   \end{split}
 \]
where the constant \[ C\!=\!\min\!\left\{\!\{ \frac{\gamma_q(\rho)}{2}\}_{q=1}^{Q}, \Big(\frac{\gamma_0(\rho)}{2}\!\!-\!\!\frac{\lambda_{\rm max}^2(\mathbf{H}^{H} \mathbf{H})}{\rho}\Big)\! \right\},\]\\
and $\gamma_q(\rho)=4^{q-1}\rho-\alpha_q$, $\gamma_0(\rho)=\rho + \lambda_{\rm min}(\mathbf{H}^{H} \mathbf{H})$.
\end{theorem}

\subsection{Computational cost}\label{complexity-sec}

  The overall {computational complexity\footnotemark} of the PS-ADMM detection algorithm consists of two parts: the first part, which is independent of the number of iterations, is required to compute the $\mathbf{x}_0 $ update of PS-ADMM in \eqref{eq: solution X0}; then, it needs to be calculated only once when detecting each transmitted symbol vector. The first part of the calculations is performed in three steps: first, the multiplication of the $U\times B$ matrix $\mathbf{H}^{H}$ by the $B\times U$ matrix $\mathbf{H}$; second, the computation of the inversion of the regularized Gramian matrix $\mathbf{H}^{H} \mathbf{H}+\rho\mathbf{I}$; and third, the computation of the  $U\times B$ matrix $\mathbf{H}^{H}$ by the $B\times 1$ vector $\mathbf{y}$ to obtain the matched-filter vector $\mathbf{H}^{H}\mathbf{y}$. These steps require $\frac{1}{2}BU^2$, $\frac{1}{3}U^3$ and $BU$ complex multiplications, respectively.
  The second part, which is iteration dependent, needs to be repeated every iteration in two steps: first, the ${Q}$ scalar multiplications by the $U$ vectors in \eqref{eq: solution_Xq}; second, a multiplication of the $U\times U$ matrix by the $U\times 1$ vector and the ${Q}$ scalar multiplications by the $U\times 1$ vectors in \eqref{eq: solution X0}. These steps require $\frac{1}{2}QU$, $U^2+\frac{1}{2}QU$ complex multiplications, respectively. Combining this result with Theorem \ref{iteration complexity}, we conclude that the total computational cost to attain an $\epsilon$-optimal solution is $\frac{1}{3}U^3+\frac{1}{2}BU^2+BU+K(U^2+QU)$, where the number of iterations $K=\frac{1}{C\epsilon}\Big(L_{\rho}(\{\mathbf{x}^{1}_q\}_{q=1}^{Q}, \mathbf{x}_0^{1}, \mathbf{y}^{1}) - \big(\frac{1}{2}\Vert\mathbf{r}-\mathbf{H}\mathbf{x}_0^* \Vert_2 ^{2} - \sum_{q=1}^{Q}\frac{\alpha_q}{2} \Vert\mathbf{x}_q^* \Vert_2 ^{2}\big)\!\Big)$. Since $K \ll B$ for massive MIMO detection, the computational complexity of the PS-ADMM mainly lies in matrix multiplication and inversion computations, which is comparable to that of the linear detector.

\footnotetext{The computational complexity is measured by the number of complex-valued multiplications for $K$ iterations. It should be noted that a complex-valued multiplication can be implemented via three real-valued multiplications.}

\begin{figure*}[htpb]
\subfigure[$B=128,U=16$ for 4-QAM;\ $\alpha_1=80,\rho=120$.]{
    \begin{minipage}{8.5cm}
    \centering
        \includegraphics[width=3.5in,height=2.7in]{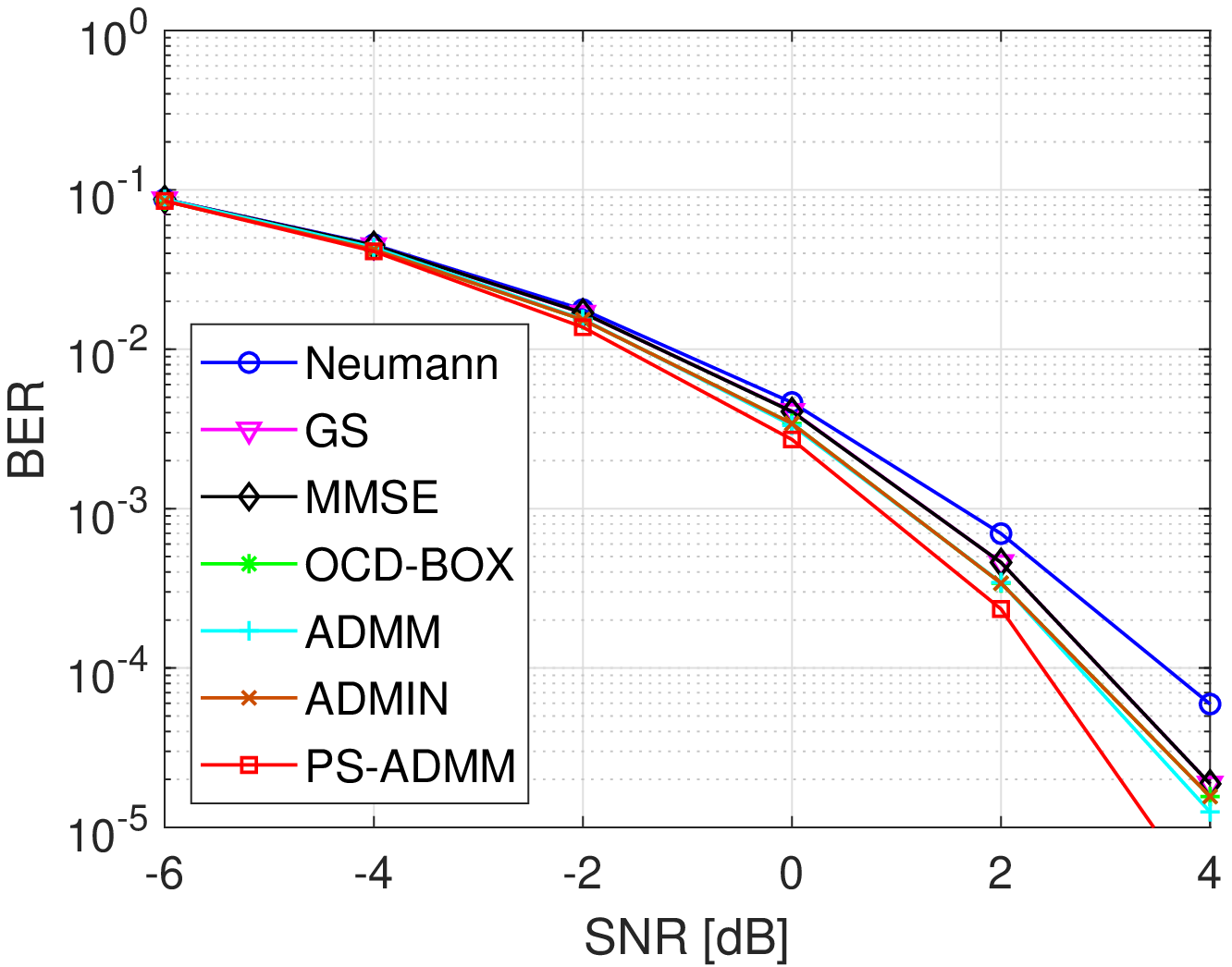}
            \label{ber-all-QPSK-128x16}
    \end{minipage}
    }
 \subfigure[$B=128,U=32$ for 4-QAM;\ $\alpha_1=80,\rho=120$.]{
    \begin{minipage}{8.5cm}
    \centering
        \includegraphics[width=3.5in,height=2.7in]{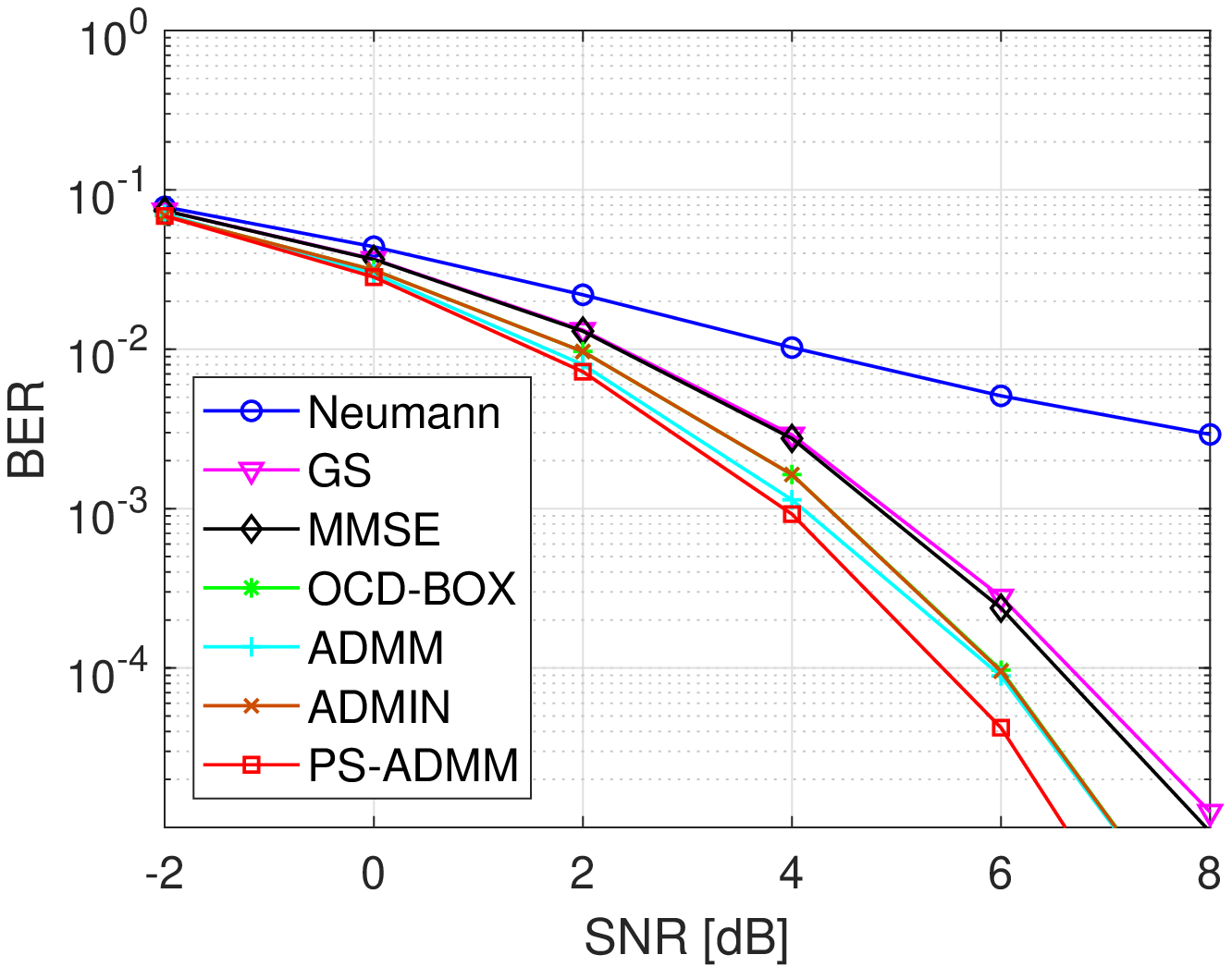}
            \label{ber-all-QPSK-128x32}
    \end{minipage}
    }

\subfigure[$B=128,U=64$ for 4-QAM;\ $\alpha_1=80,\rho=120$.]{
    \begin{minipage}{8.5cm}
    \centering
        \includegraphics[width=3.5in,height=2.7in]{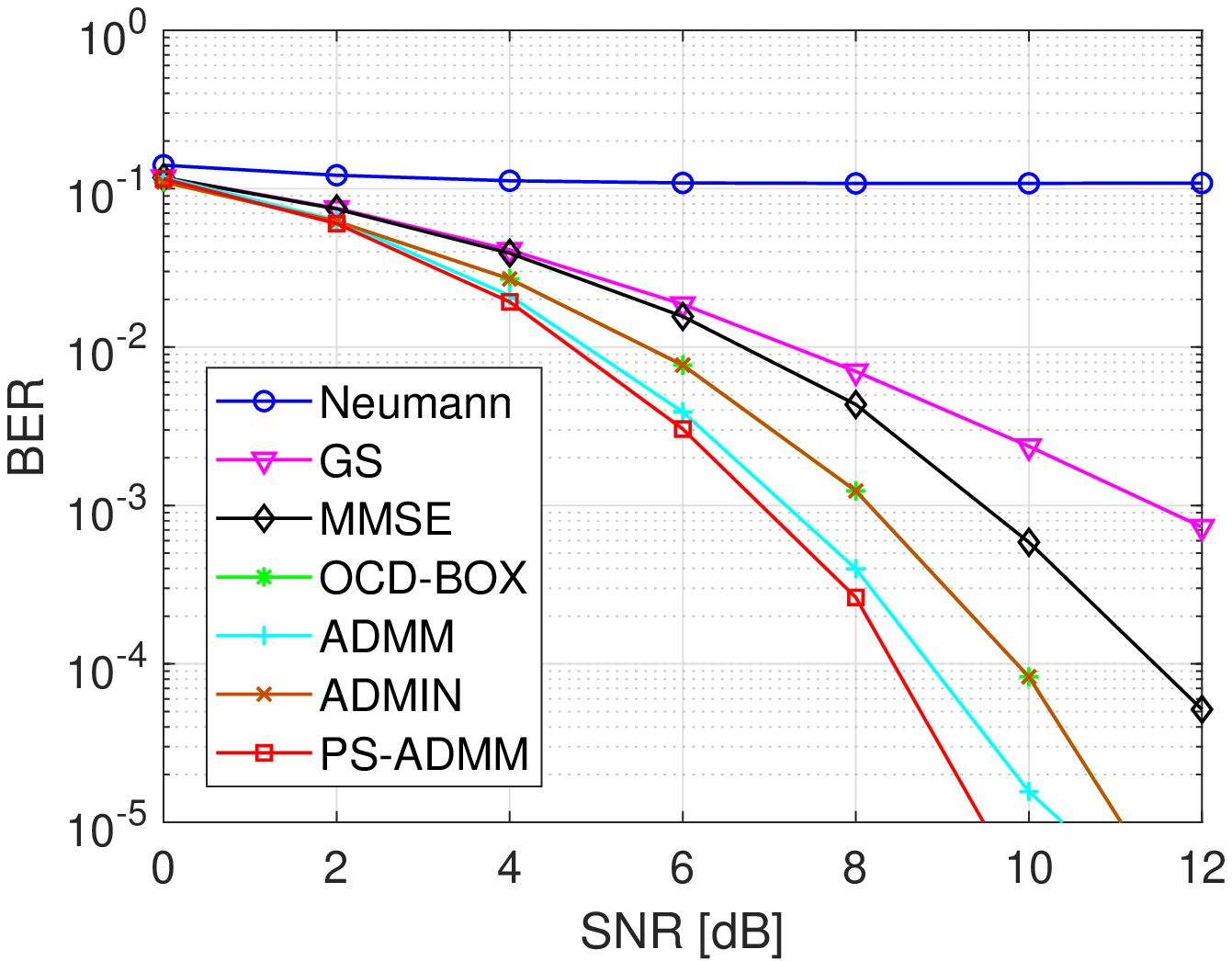}
            \label{ber-all-QPSK-128x64}
    \end{minipage}
    }
 \subfigure[$B=128,U=128$ for 4-QAM;\ $\alpha_1=80,\rho=120$.]{
    \begin{minipage}{8.5cm}
    \centering
        \includegraphics[width=3.5in,height=2.7in]{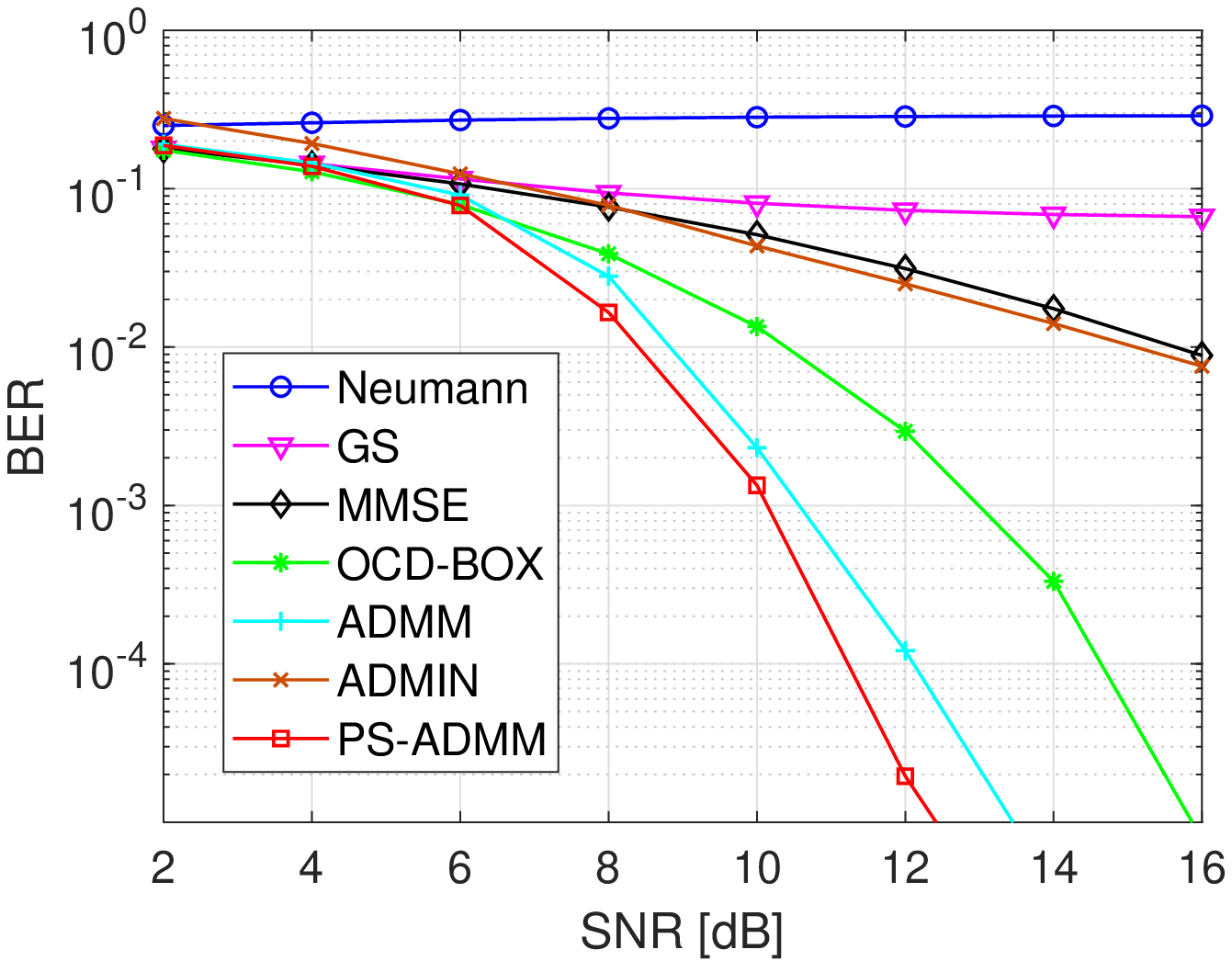}
            \label{ber-all-QPSK-128x128}
    \end{minipage}
    }
 \subfigure[$B=128,U=16$ for 16-QAM;\ $\alpha_1=8,\alpha_2=30,\rho=16$.]{
    \begin{minipage}{8.5cm}
    \centering
        \includegraphics[width=3.5in,height=2.7in]{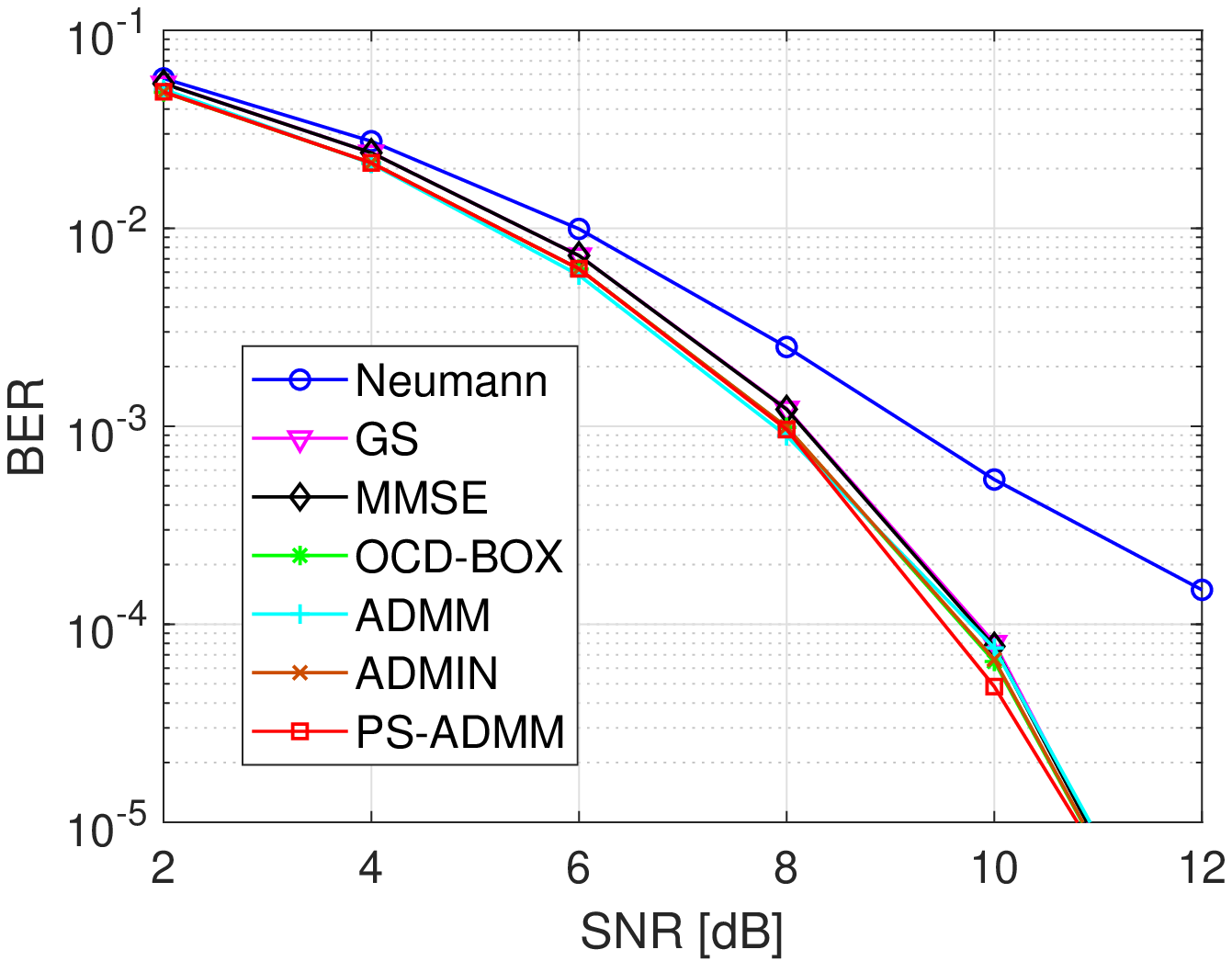}
            \label{ber-all-16qam-128x16}
    \end{minipage}
    }
 \subfigure[$B=128,U=32$ for 16-QAM;\ $\alpha_1=9,\alpha_2=40,\rho=20$.]{
    \begin{minipage}{8.5cm}
    \centering
        \includegraphics[width=3.5in,height=2.7in]{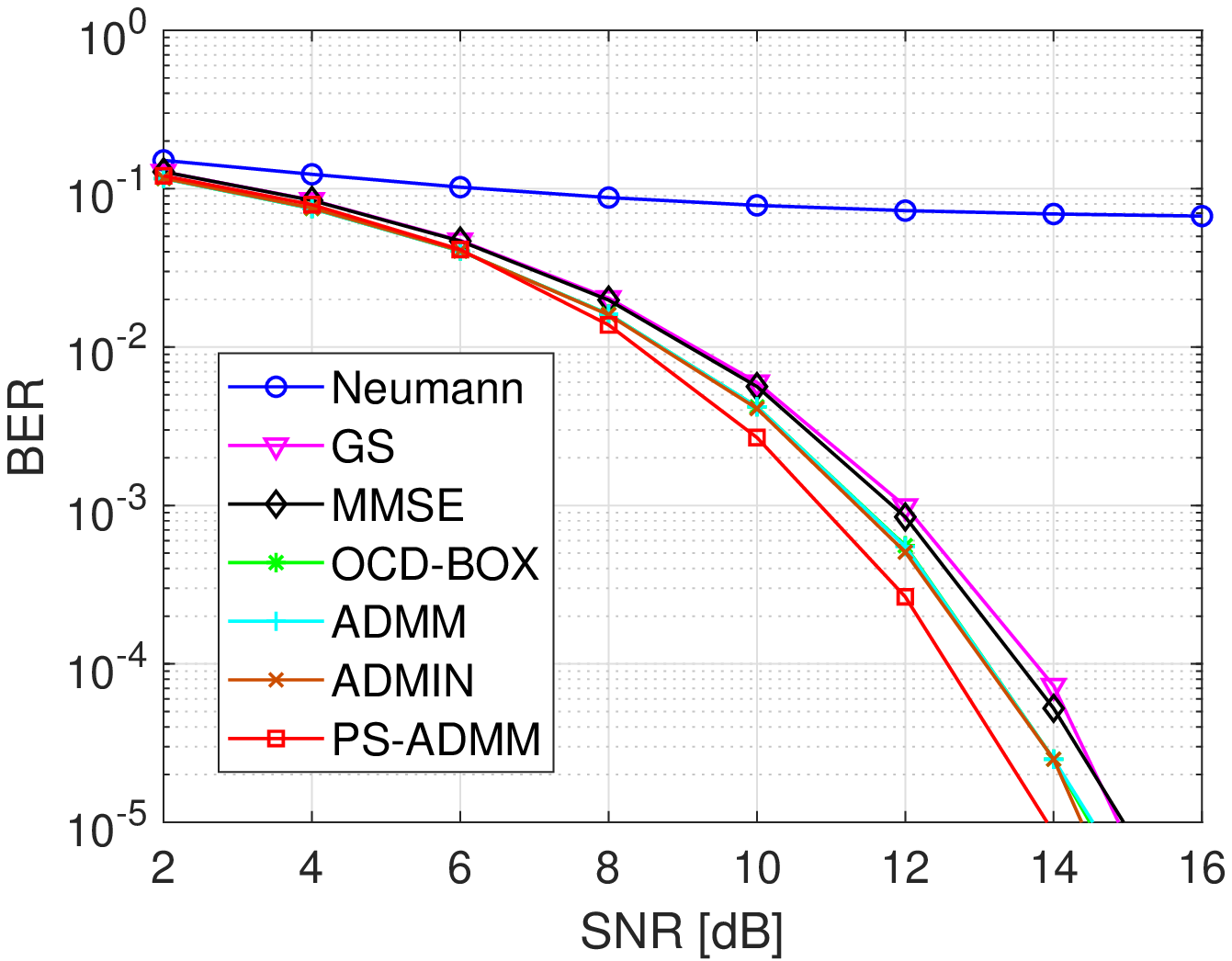}
            \label{ber-all-16qam-128x32}
    \end{minipage}
    }

\end{figure*}
\addtocounter{figure}{-1}
\begin{figure*}
\addtocounter{figure}{1}

 \subfigure[$B=128,U=64$ for 16-QAM;\ $\alpha_1=12,\alpha_2=64,\rho=20$.]{
    \begin{minipage}{8.5cm}
    \centering
        \includegraphics[width=3.5in,height=2.7in]{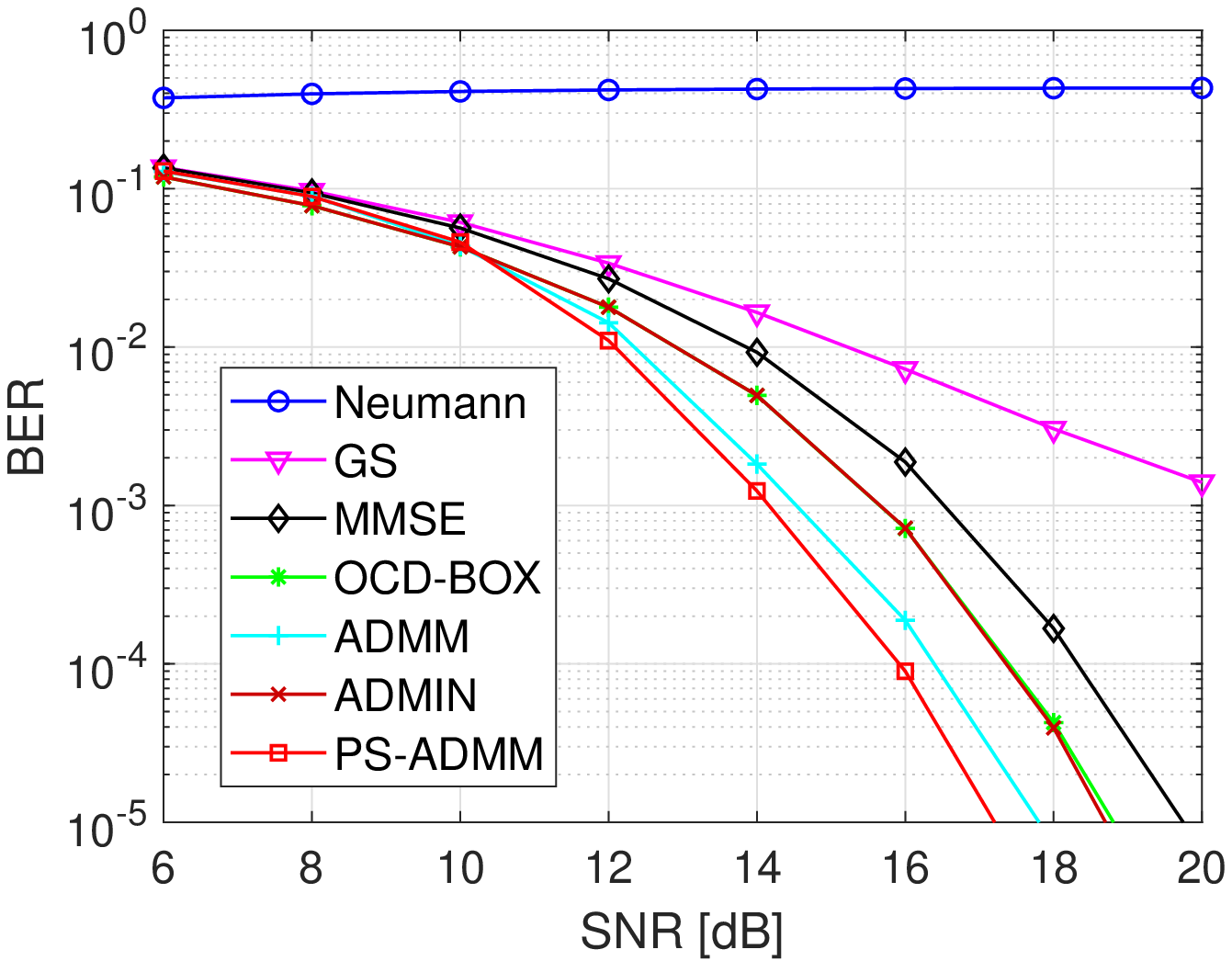}
            \label{ber-all-16qam-128x64}
    \end{minipage}
    }
 \subfigure[$B=128,U=128$ for 16-QAM;\ $\alpha_1=10,\alpha_2=60,\rho=16$.]{
    \begin{minipage}{8.5cm}
    \centering
        \includegraphics[width=3.5in,height=2.7in]{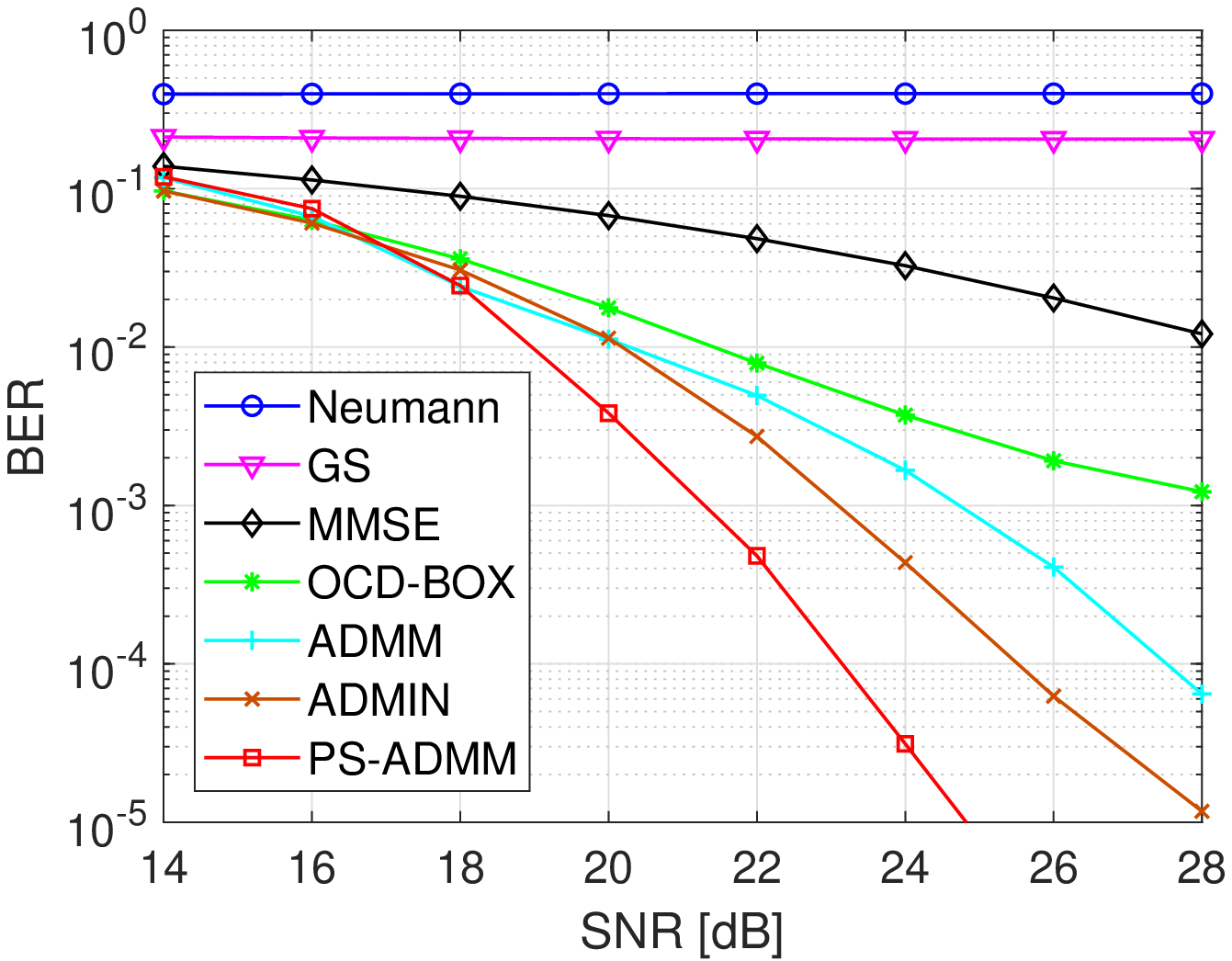}
            \label{ber-all-16qam-128x128}
    \end{minipage}
    }
\subfigure[$B=128,U=16$ for 64-QAM;\ $\alpha_1=22$,$\alpha_2=17$,$\alpha_3=95$,$\rho=96$.]{
    \begin{minipage}{8.5cm}
    \centering
        \includegraphics[width=3.5in,height=2.7in]{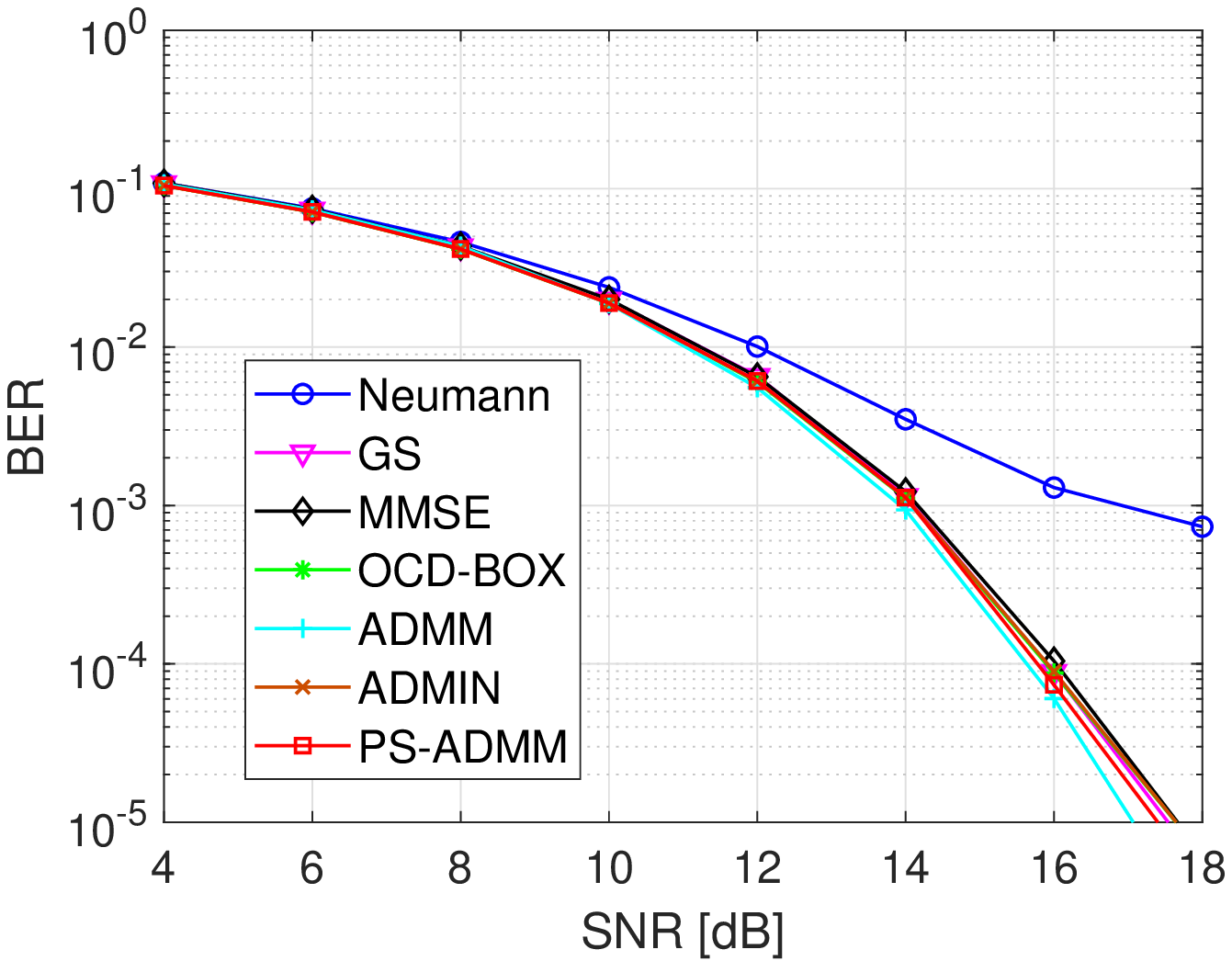}
            \label{ber-all-64qam-128x16}
    \end{minipage}
    }
\subfigure[$B=128,U=32$ for 64-QAM;\ $\alpha_1=2$,$\alpha_2=2$,$\alpha_3=10.5$,$\rho=9$.]{
    \begin{minipage}{8.5cm}
    \centering
        \includegraphics[width=3.5in,height=2.7in]{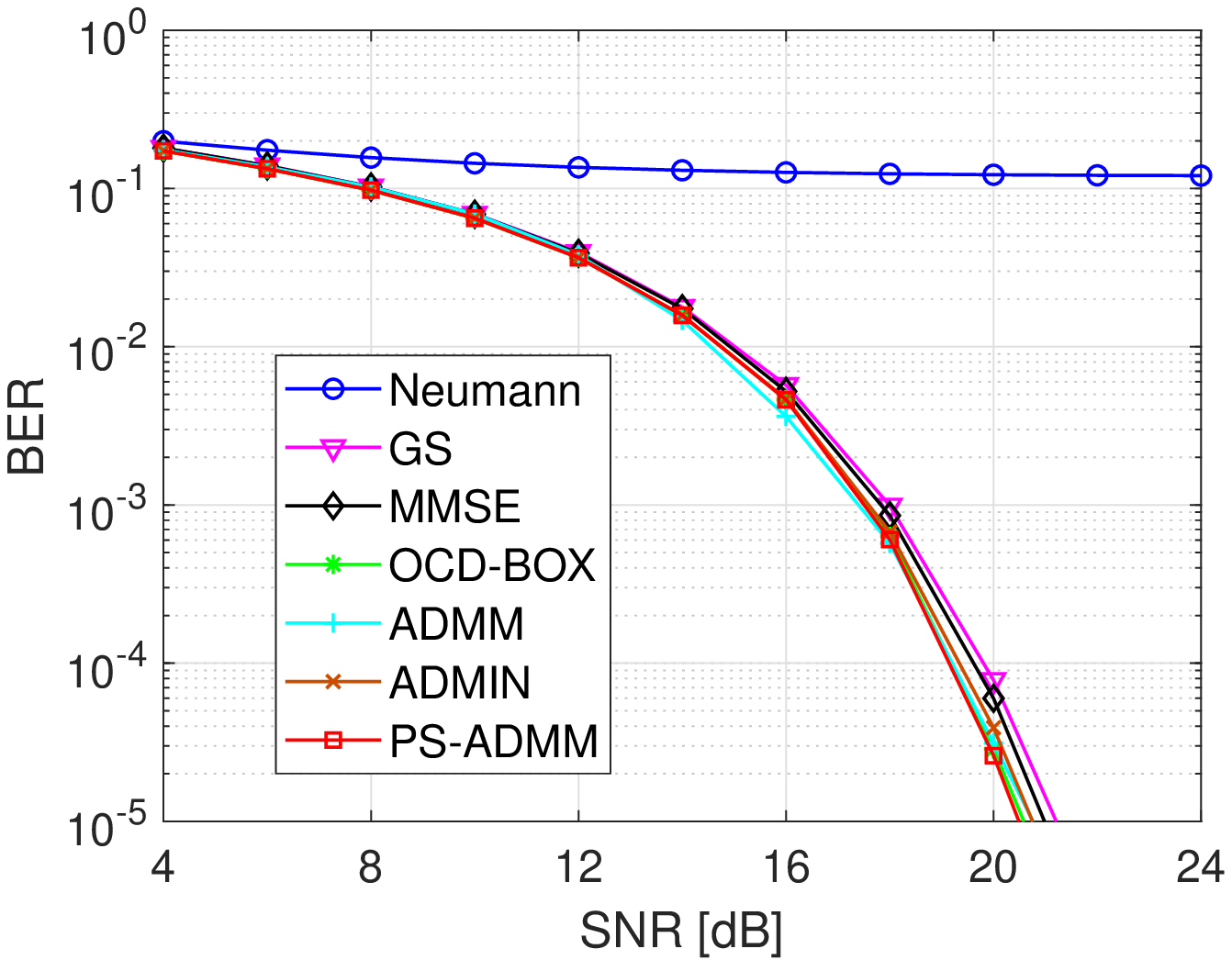}
            \label{ber-all-64qam-128x32}
    \end{minipage}
    }
\subfigure[$B=128,U=64$ for 64-QAM;\ $\alpha_1=22$, $\alpha_2=22.5$, $\alpha_3=85$, $\rho=44$.]{
    \begin{minipage}{8.5cm}
    \centering
        \includegraphics[width=3.5in,height=2.7in]{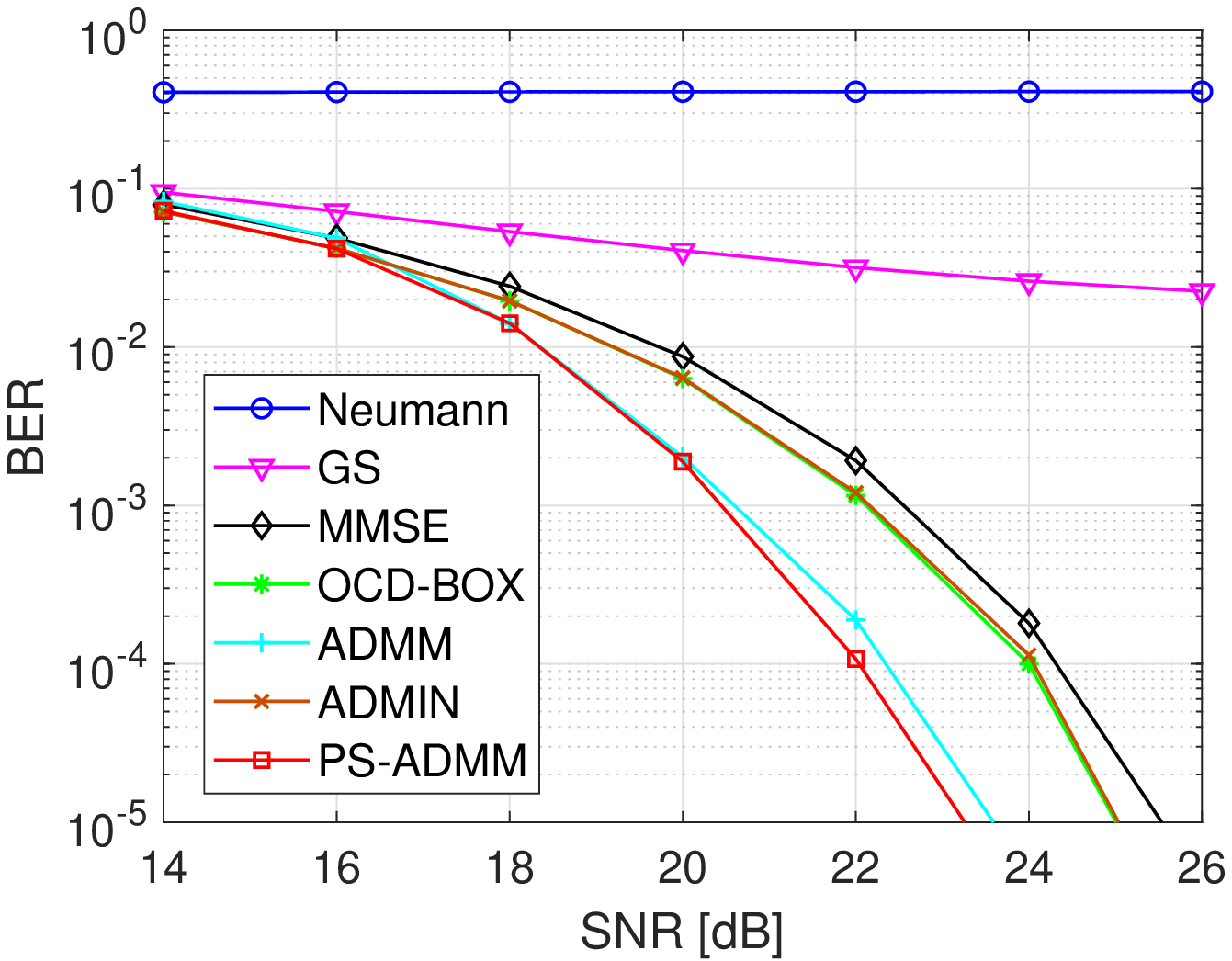}
            \label{ber-all-64qam-128x64}
    \end{minipage}
    }
  \subfigure[$B=128,U=128$ for 64-QAM;\ $\alpha_1=2.75$, $\alpha_2=2.25$, $\alpha_3=10.5$, $\rho=5$.]{
    \begin{minipage}{8.5cm}
    \centering
        \includegraphics[width=3.5in,height=2.7in]{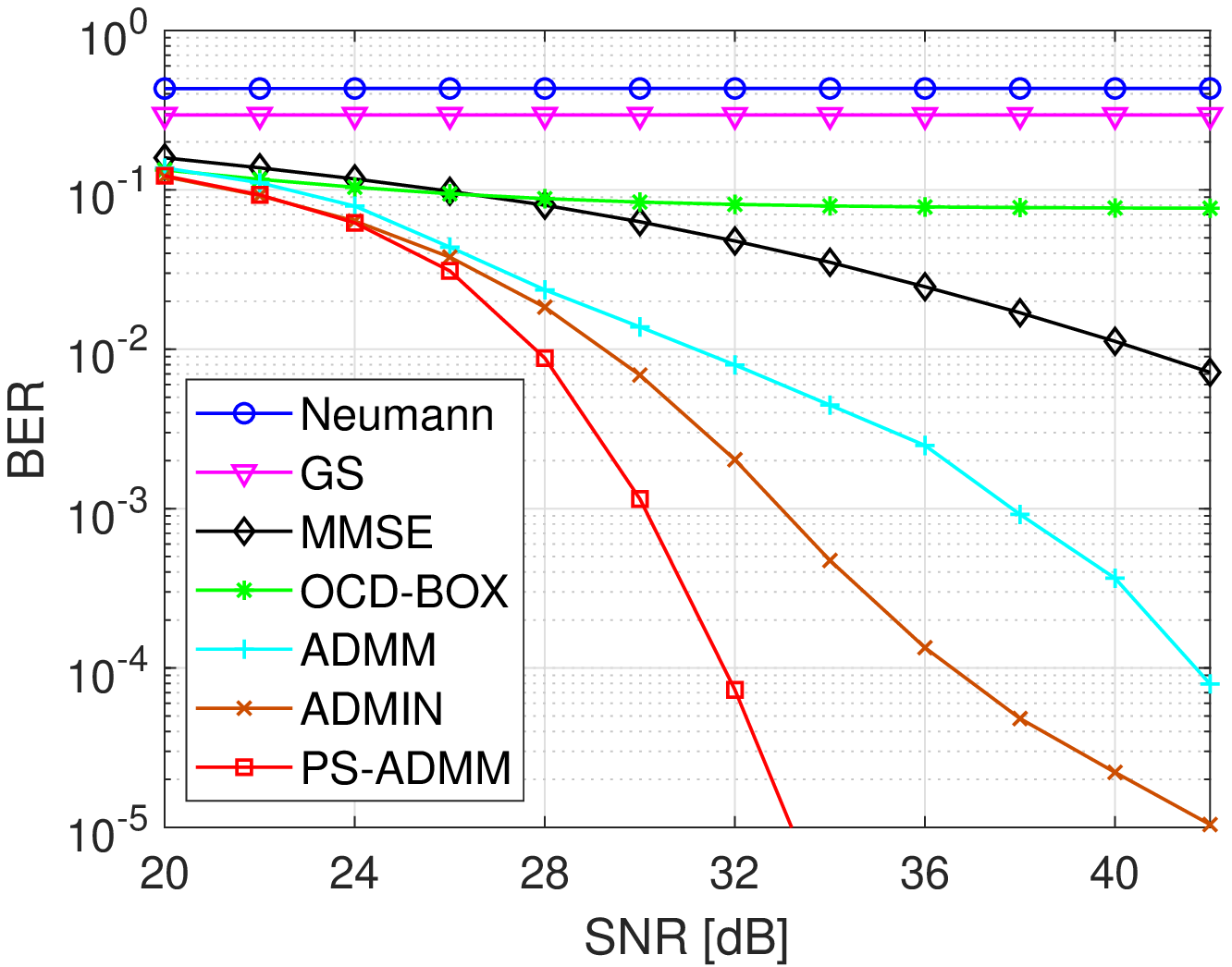}
            \label{ber-all-64qam-128x128}
    \end{minipage}
    }
    \centering
    \caption{Comparisons of BER performance using various massive MIMO detectors.}
    \label{ber_allalgo}
 \end{figure*}

\section{Simulation results}\label{sec:Simulation results}

In this section, numerical results are presented to show the effectiveness of the proposed PS-ADMM detector.
Specifically, in Section \ref{simulation-result-performance} we demonstrate the BER performance of the PS-ADMM detector compared with several state-of-the-art detectors.
In Section \ref{simulation-parameter-setting}, we focus on analyzing the impact of parameters $\rho$, $\{\alpha_q\}_{q=1}^{Q}$, and $K$ on the performance of the PS-ADMM detector.

Throughout this section, we show simulation results for uncoded and hard decision-based signal detection with the i.i.d. Rayleigh fading channel in different $B \times U \;(B \ge U)$ multiuser massive MIMO systems. The modulation schemes of 4-QAM (QPSK), 16-QAM, and 64-QAM are considered. We assume that perfect knowledge of the channel state information is known at the receiver side. For a fair comparison, all the algorithms are implemented using Matlab 2019a/Windows 7 environment on a computer with 3.7GHz Intel i3-6100$\times$2 CPU and 16GB RAM.

\subsection{BER performance}\label{simulation-result-performance}
In this subsection, the BER performance of the proposed PS-ADMM detector was evaluated and compared with conventional and state-of-the-art MIMO detectors by numerical simulations, which are the classical MMSE detector, Neumann detector \cite{wu2014large}, GS detector \cite{wu2016efficient}, OCD-BOX detector \cite{wu2016high}, and two ADMM-based detectors ADMM and ADMIN in \cite{7526551} and \cite{shahabuddin2021admm} respectively. The termination criterium is that iteration number $K$ reaches $30$ or the residual in \eqref{residual_k} is less than $10^{-5}$. The data points plotted in all BER curves are averaged over 1000 Monte-Carlo trials.

 Fig.\,\ref{ber_allalgo} shows BER performance of considered detectors for 4-QAM, 16-QAM, and 64-QAM modulation with a different number of antennas in the transmitter of the massive MIMO systems.
 In Fig.\,\ref{ber_allalgo}, one can see that BER curves of all detectors have a similar changing trend at low SNRs that continues to drop in a waterfall manner in relatively high SNR regions. We observe that the PS-ADMM detector achieves the best BER performance.
 Observing Fig.\,\ref{ber-all-QPSK-128x16}-\ref{ber-all-QPSK-128x64}, \ref{ber-all-16qam-128x16}-\ref{ber-all-16qam-128x64}, and \ref{ber-all-64qam-128x16}-\ref{ber-all-64qam-128x64}, we can find that all of the detectors have comparable BER performance when the BS-to-user-antenna ratio is more than two; only the approximate matrix inversion algorithms such as Neumann and GS suffer from performance loss.
 In Fig.\,\ref{ber-all-QPSK-128x128},\,\ref{ber-all-16qam-128x128}, and \,\ref{ber-all-64qam-128x128}, for the more challenging $128\times128$ massive MIMO systems, one can also see that the PS-ADMM detector outperforms other detectors. From the presented simulation results, we can see clearly that the proposed PS-ADMM detector can achieve better BER performance than the-state-of-the-art ones and the gap increases when the ratio of the BS's antenna number to the user number approaches one.

\begin{figure}[tp]
   \subfigure[BER performance vs. $\rho$.]{
    \begin{minipage}{9cm}
    \centering
        \includegraphics[width=3.5in,height=2.7in]{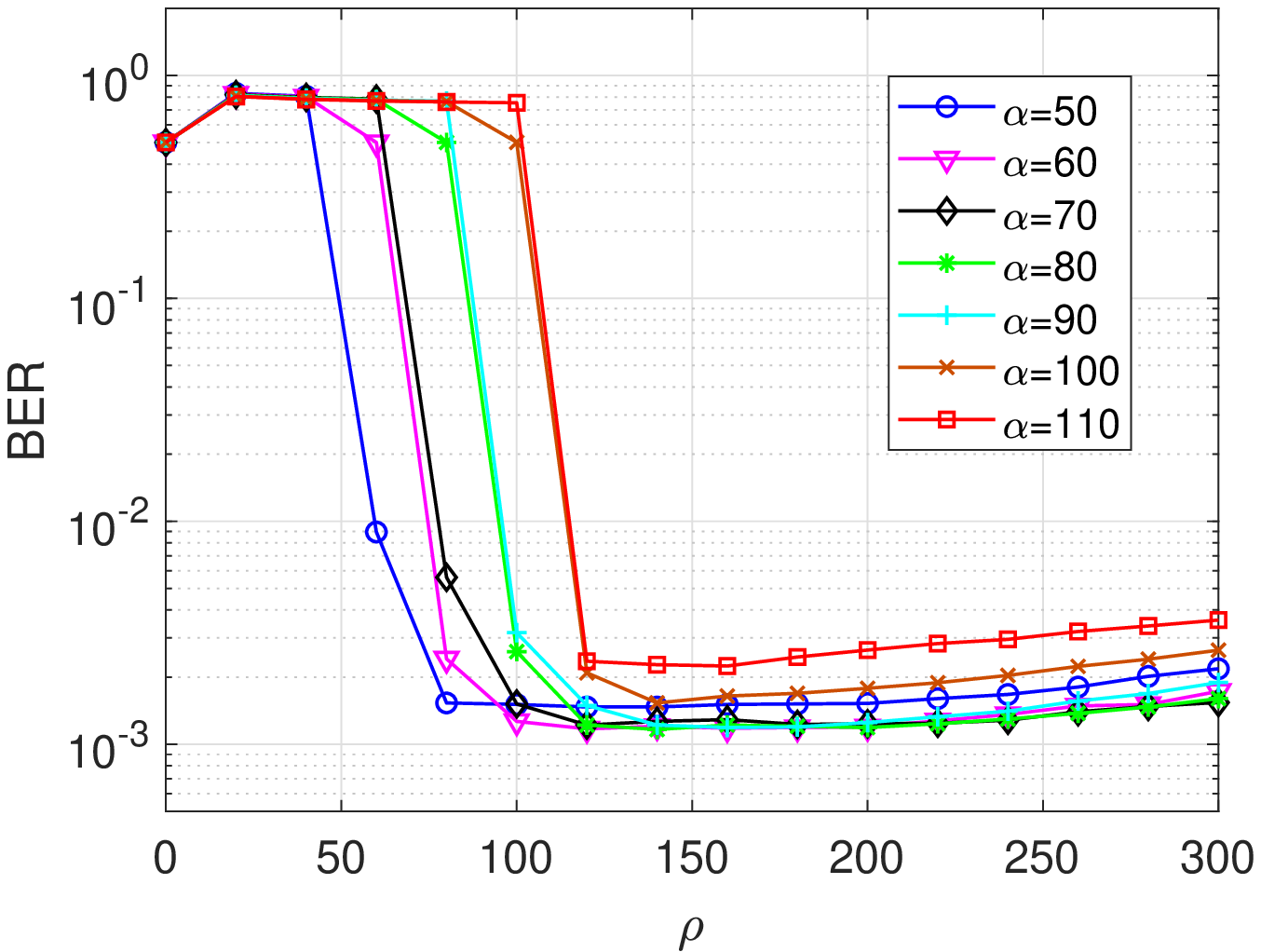}
            \label{ber_rho_4qam}
    \end{minipage}
    }
   \subfigure[BER performance vs. $\alpha$.]{
    \begin{minipage}{9cm}
    \centering
        \includegraphics[width=3.5in,height=2.7in]{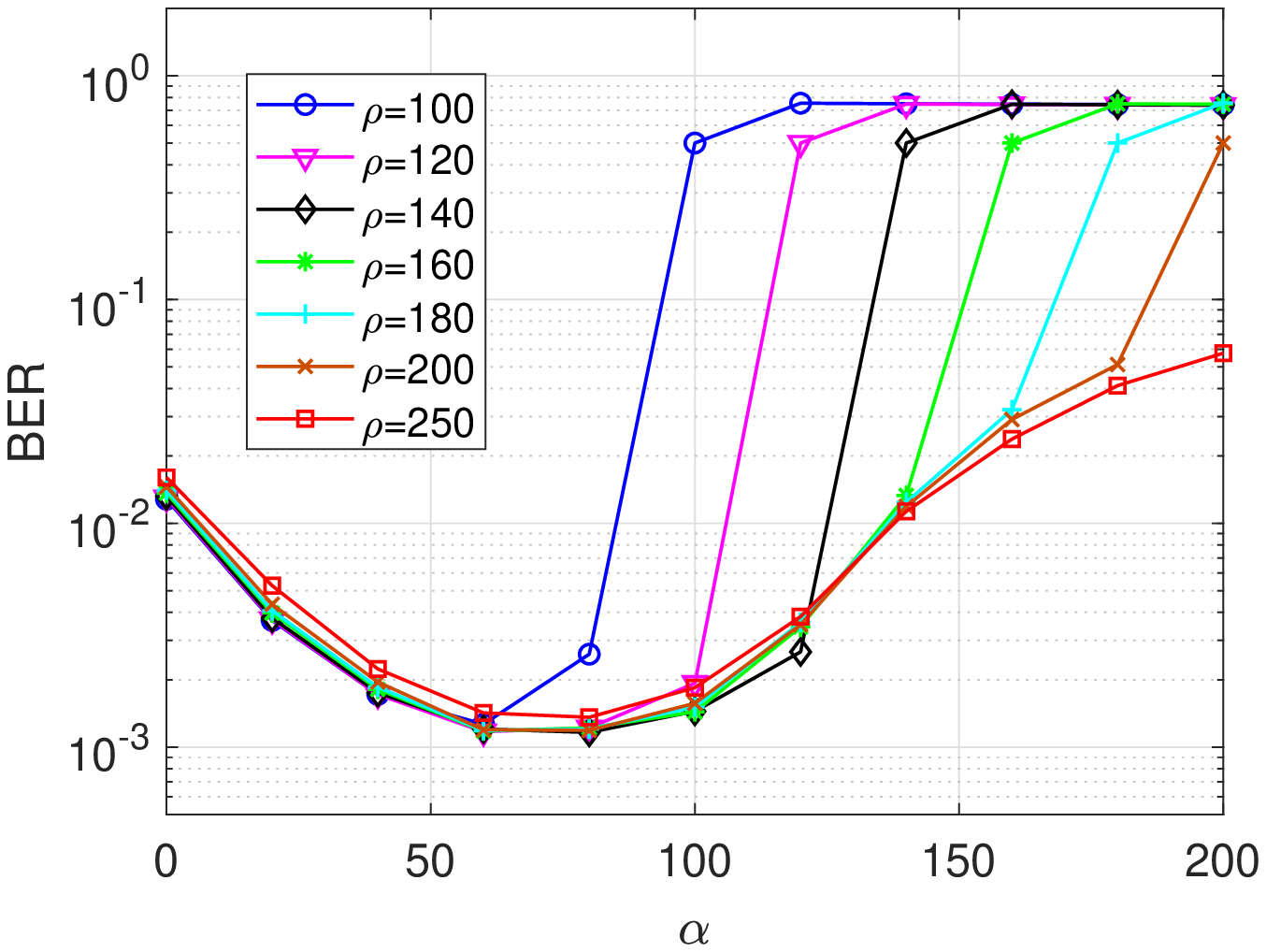}
            \label{ber_alpha_4qam}
    \end{minipage}
    }
 \centering
 \caption{The impact of $\rho$ and $\alpha$ on BER performance of the PS-ADMM detector.}
 \label{impact_BER}
\end{figure}
\begin{figure}[tp]
  \subfigure[Iteration number vs. $\rho$.]
    {
    \begin{minipage}{9cm}
    \centering
        \includegraphics[width=3.5in,height=2.7in]{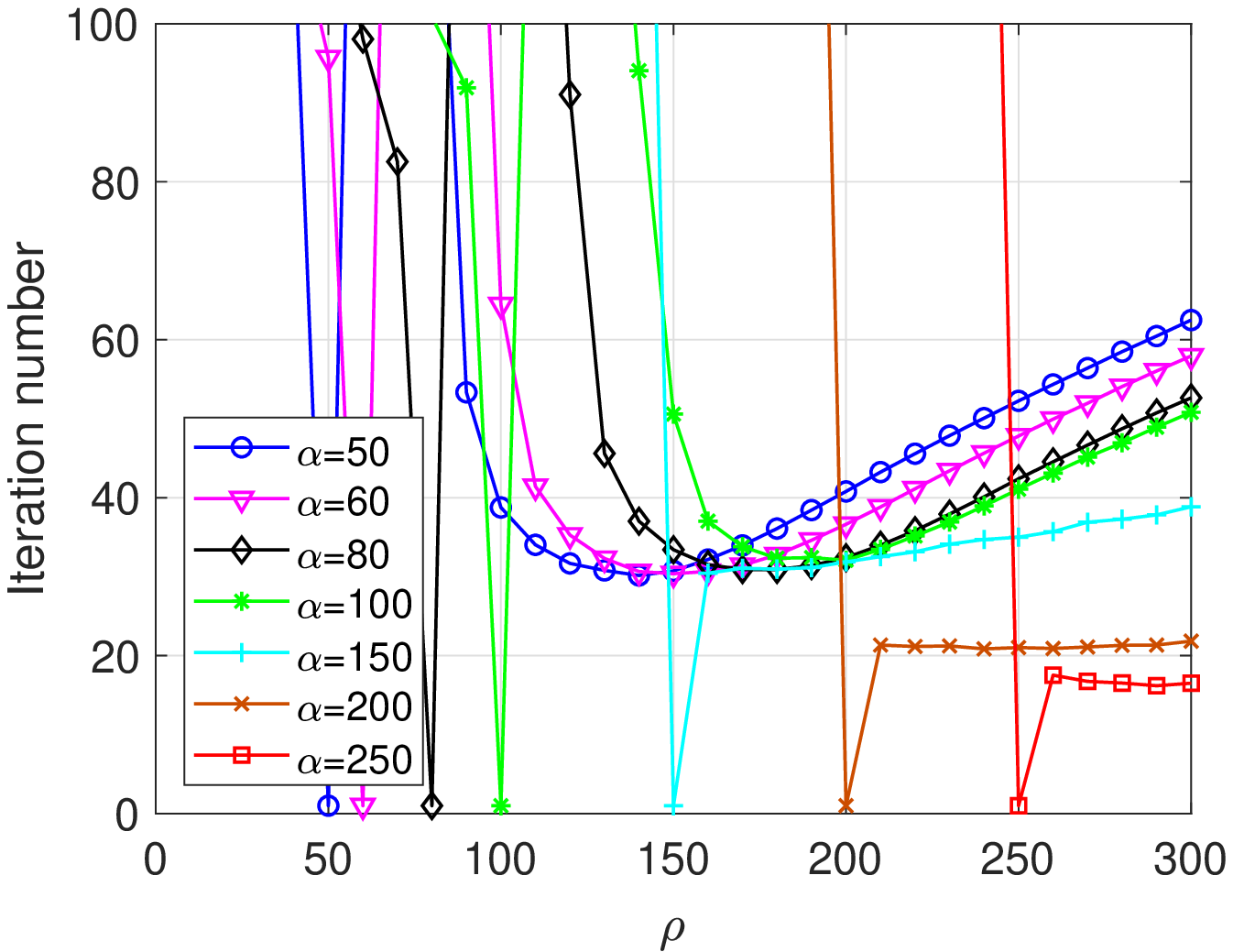}
            \label{iter_rho_4qam}
    \end{minipage}
    }
  \subfigure[Iteration number vs. $\alpha$.]
    {
    \begin{minipage}{9cm}
    \centering
        \includegraphics[width=3.5in,height=2.7in]{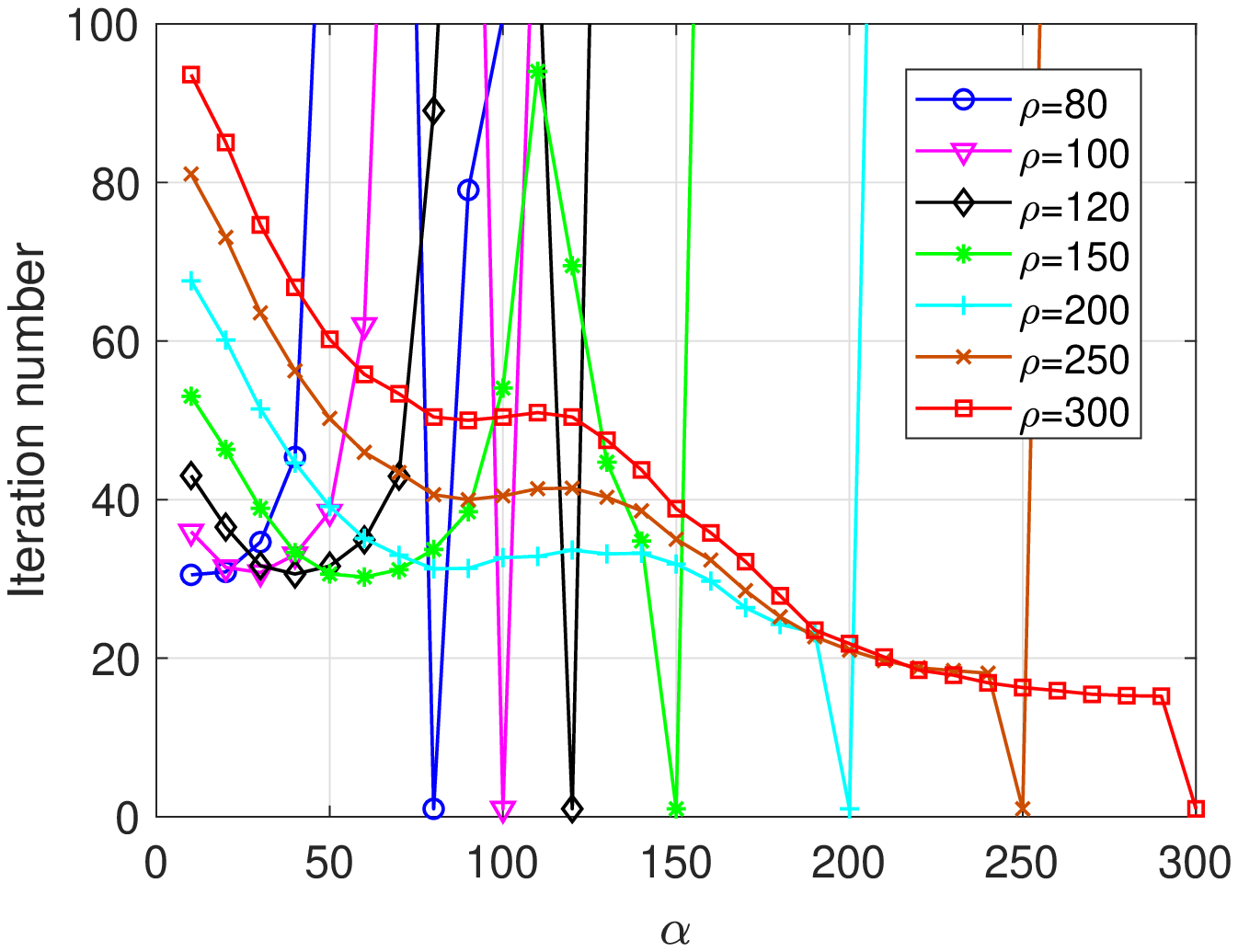}
            \label{iter_alpha_4qam}
    \end{minipage}
    }
 \centering
 \caption{The impact of $\rho$ and $\alpha$ on convergence performance of the PS-ADMM detector.}
 \label{impact_conv}
\end{figure}

 \subsection{Choice of Parameters} \label{simulation-parameter-setting}

In this subsection, we show that the proper parameters $\rho$ and $\alpha$ can achieve lower BER performance and speed up convergence of the proposed PS-ADMM detector. The considered modulation scheme is 4-QAM and the simulation parameters are $B=128$, $U=128$, and ${\rm SNR}=10{\rm dB}$.

In Fig.\,\ref{impact_BER}, it shows the effects on BER performance when the different values of the penalty parameters $\rho$ and $\alpha$ are chosen. From the figure, one can have the following observations: first, both $\rho$ and $\alpha$ can affect BER performance of the proposed PS-ADMM decoder; second, too large or too small values of $\alpha$ and $\rho$ can worsen BER performance of the detector. For the case of the presented simulation, one can see the proper $\rho \in [120~200]$ and $\alpha \in [60~80]$ respectively. Moreover, we note that when $\alpha$ approaches zero, the PS-ADMM detector degenerates to the conventional ADMM detector with a box constraint.

In Fig.\,\ref{impact_conv}, we study the effects of the penalty parameters $\rho$ and $\alpha$ on the convergence characteristic of the proposed PS-ADMM decoder.
From the figures, one can observe that the proposed PS-ADMM algorithm can always converge with different settings of the parameters $\rho$ and $\alpha$ when $\rho>\alpha$ is satisfied, and these parameters can affect its convergence rate.
From the figure, it shows that the larger the penalty parameter $\rho$ and the value of $\alpha$ is close to $\rho$, the faster the algorithm converges, but these too large parameters $\rho$ and $\alpha$ are not a good choice for BER performance of the PS-ADMM detector. There is no need to sacrifice a lot of BER performance just to reduce a dozen iterations. We can observe that the optimal value of $\rho$ and $\alpha$ in terms of convergence rate agrees with the optimal value of $\rho$ and $\alpha$ in terms of BER performance when iteration number reaches $30$.

In Fig.\,\ref{impact_conv_K}, not only can the impact of $\rho$ and $\alpha$ on convergence performance be observed, it can also see that the proposed PS-ADMM algorithm can converge within a few tens of iterations to converge to modest accuracy solutions, which is promising for large-scale MIMO detection scenarios.

Specifically, for higher order modulations, proper BER and convergence performance can be obtained only when $4^{q-1}\rho>\alpha_q, q = 1,\cdots,Q$ is satisfied.

\begin{figure}[tp]
  \subfigure
    {
    \begin{minipage}{9cm}
    \centering
        \includegraphics[width=3.5in,height=2.7in]{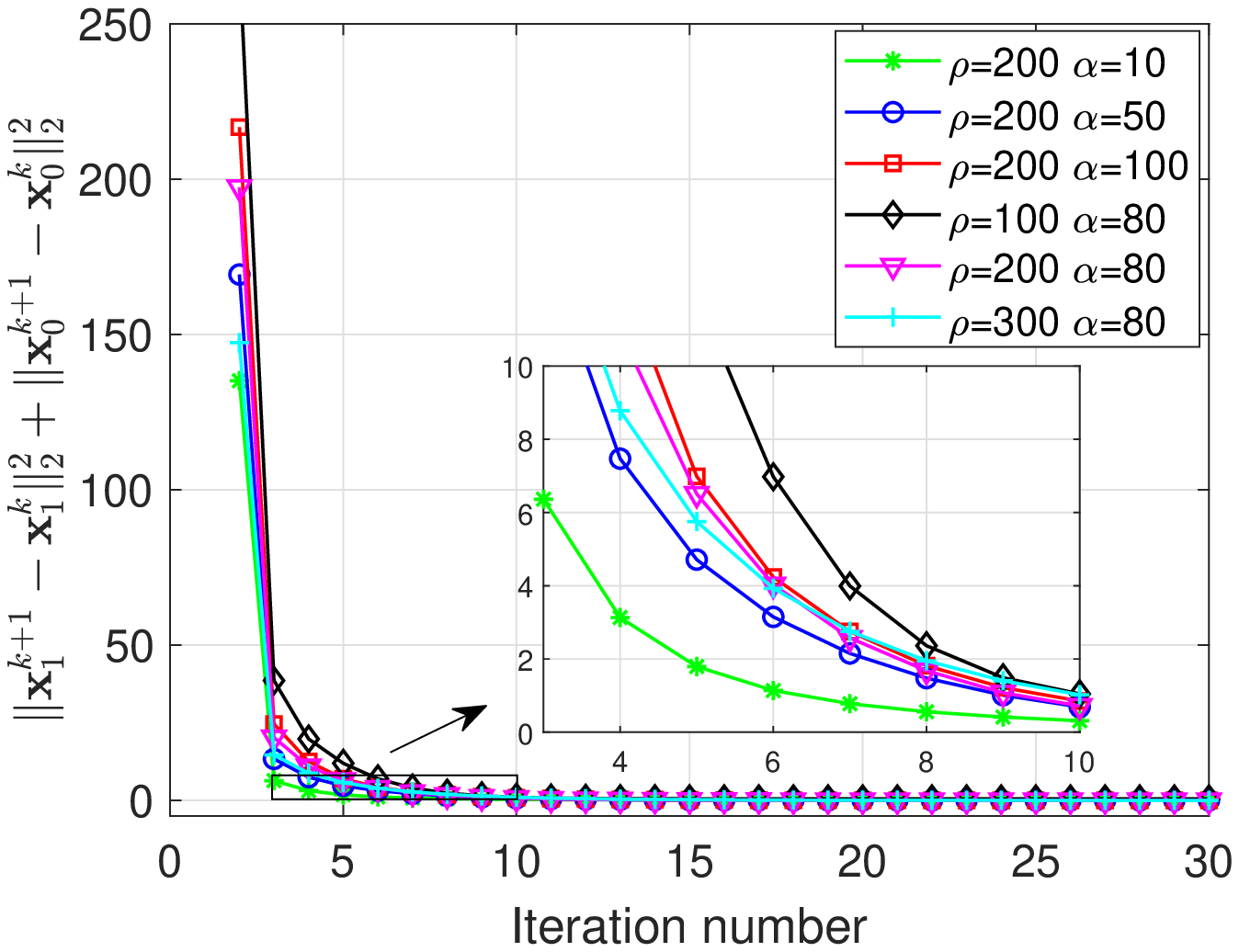}
            \label{res_iter_4qam}
    \end{minipage}
    }
   \centering
 \caption{The impact of the maximum iteration number $K$ on convergence performance of the PS-ADMM detector.}
 \label{impact_conv_K}
\end{figure}

\section{Conclusion}
\label{sec:Conclusion}

In this paper, we proposed a new MIMO detector for high-order QAM modulation signals via the PS-ADMM approach. We prove that the proposed PS-ADMM approach is theoretically-guaranteed convergent under some wild conditions. Compared with several state-of-the-art MIMO detectors, the proposed PS-ADMM detector has competitive BER performance and cheap computational complexity, especially when the ratio of the antenna number in the BS and user number is close to one. Since channel coding and soft decision-based MIMO detection is essentially used in practical systems, it would be meaningful to further study how the proposed PS-ADMM approach can be applied to such scenarios. In addition, how to choose the optimal penalty parameter is also an interesting research topic.

\appendices

\section{Proof of Lemmas \ref{lemma:z1}--\ref{lemma:L_bounded1}}\label{lemma1-3}
Before proving convergence of the proposed PS-ADMM algorithm, we give several lemmas and their proofs as follows.

\begin{lemma}\label{lemma:z1}
For Algorithm 1, the following inequality holds
\begin{equation}\label{eq:y_difference1}
\|\mathbf{y}^{k+1}-\mathbf{y}^k\|_2^2 \le \lambda_{\rm max}^2(\mathbf{H}^H\mathbf{H})\|\mathbf{x}_0^{k+1}-\mathbf{x}_0^k\|_2^2.
\end{equation}
\end{lemma}
\begin{proof}
Since $\mathbf{x}_0^{k+1}$ is a minimizer of problem \eqref{eq:x0_update}, it should satisfy the following optimality condition \cite{bertsekas2009convex}
\begin{align}\label{eq:x_nabla_expression}
\nabla_{\mathbf{x}_0} \ell\left(\mathbf{x}_0^{k+1}\right)+\mathbf{y}^k+\rho (\mathbf{x}_0^{k+1}-\sum_{q=1}^{Q}2^{q-1} \mathbf{x}_q^{k+1})=0.
\end{align}
Plugging $\mathbf{y}^{k+1}$ in \eqref{eq:y_update} into the above equation, we obtain
\begin{equation}\label{eq:z_nabla_expression}
\mathbf{y}^{k+1}=-\nabla_{\mathbf{x}_0} \ell\left(\mathbf{x}_0^{k+1}\right).
\end{equation}
According to Lagrange's mean value theorem, since $\ell\left(\mathbf{x}_0\right)$ is continuous and differentiable, there exists some point $\bar{\mathbf{x}}_0$ between $\mathbf{x}_0^{k}$ and $\mathbf{x}_0^{k+1}$ which satisfies
\begin{equation}\label{mean value theorem}
\frac{\nabla_{\mathbf{x}_0} \ell\left(\mathbf{x}_0^{k+1}\right)-\nabla_{\mathbf{x}_0} \ell\left(\mathbf{x}_0^{k}\right)}{\mathbf{x}_0^{k+1}-\mathbf{x}_0^{k}}= \nabla_{\bar{\mathbf{x}}_0}^2 \ell\left(\mathbf{x}_0\right).
\end{equation}
Moreover, since $\nabla_{\bar{\mathbf{x}}_0}^2 \ell\left(\mathbf{x}_0\right) =\mathbf{H}^{H} \mathbf{H} \preceq \lambda_{\rm max}(\mathbf{H}^{H} \mathbf{H})\mathbf{I}$, we have \begin{equation}\label{lipschitz x0}
\begin{split}
\!\!\|\!\nabla_{\mathbf{x}_0} \ell\left(\mathbf{x}_0^{k+1}\right)\!\!-\!\!\nabla_{\mathbf{x}_0} \ell\left(\mathbf{x}_0^{k}\right)\!\!\|_2^2
 \!\!\le\!\!\lambda_{\rm max}^2(\mathbf{H}^{H}\!\mathbf{H}) \|\mathbf{x}_0^{k+1}\!\!\!-\!\mathbf{x}_0^{k}\|_2^2.
 \end{split}
\end{equation}
From \eqref{lipschitz x0}, we can see that $\nabla_{\mathbf{x}_0} \ell\left(\mathbf{x}_0\right)$ is Lipschitz continuous with constant $\lambda_{\rm max}(\mathbf{H}^{H} \mathbf{H})$.
Plugging \eqref{eq:z_nabla_expression} into LHS of equation \eqref{lipschitz x0}, we can obtain
\begin{align*}
\begin{split}
\|\mathbf{y}^{k+1}-\mathbf{y}^k\|_2^2
=&\|\nabla_{\mathbf{x}_0} \ell\left(\mathbf{x}_0^{k+1}\right)-\nabla_{\mathbf{x}_0} \ell\left(\mathbf{x}_0^{k}\right)\|_2^2 \\
\le & \lambda_{\rm max}^2(\mathbf{H}^{H} \mathbf{H}) \|\mathbf{x}_0^{k+1}-\mathbf{x}_0^{k}\|_2^2.
\end{split}
\end{align*}
This completes the proof.
$\hfill\blacksquare$
\end{proof}

\begin{lemma}\label{lemma:L_difference}
Let $\alpha_q$ and $\rho$ satisfy $4^{q-1}\rho>\alpha_q$, $\forall q\in\{1,\dotsb, Q\}$. Then, for Algorithm 1, we have the following inequality
{\setlength\abovedisplayskip{2pt}
 \setlength\belowdisplayskip{2pt}
 \setlength\jot{2pt}
\begin{equation}\label{eq:L_difference}
\begin{split}
\!\!&L_{\rho}\left(\{\mathbf{x}_q^{k+1}\}_{q=1}^{Q}\!, \mathbf{x}_0^{k+1}\!, \mathbf{y}^{k+1}\right)\!-\!L_{\rho}\left (\{\mathbf{x}_q^{k}\}_{q=1}^{Q}, \mathbf{x}_0^{k}\!, \mathbf{y}^{k} \right)\\
\!\!\le&\!\!\sum_{q=1}^{Q}\!-\!\frac{\gamma_q(\rho)}{2}\|\mathbf{x}_q^{k+1}\!\!-\!\!\mathbf{x}_{q}^{k}\|_2^2 \!\!-\!\!\Big(\frac{\gamma_0(\rho)}{2}\!\!-\!\!\frac{\lambda_{\rm max}^2(\mathbf{H}^{H} \mathbf{H})}{\rho}\Big)\|\mathbf{x}_0^{k+1}\!\!-\!\!\mathbf{x}_0^k\|_2^2,
\end{split}
\end{equation}
where} $\gamma_q(\rho)=4^{q-1}\rho-\alpha_q$ and $\gamma_0(\rho)=\rho + \lambda_{\rm min}(\mathbf{H}^{H} \mathbf{H})$.
\end{lemma}
\begin{proof}
We split LHS of the inequality \eqref{eq:L_difference} into two terms
{\setlength\abovedisplayskip{2pt}
 \setlength\belowdisplayskip{2pt}
 \setlength\jot{2pt}
\begin{equation}\label{eq:successive_L1}
\begin{split}
&L_{\rho}(\{\mathbf{x}_q^{k+1}\}_{q=1}^Q, \mathbf{x}_0^{k+1}, \mathbf{y}^{k+1})-L_{\rho}(\{\mathbf{x}^{k}_q\}_{q=1}^Q, \mathbf{x}_0^{k}, \mathbf{y}^{k})\\
&=\!\underbrace{\left(L_{\rho}(\{\mathbf{x}_q^{k+1}\}_{q=1}^Q, \mathbf{x}_0^{k+1},\! \mathbf{y}^{k+1})\!-\!L_{\rho}(\{\mathbf{x}_q^{k+1}\}_{q=1}^Q,\! \mathbf{x}_0^{k+1},\! \mathbf{y}^{k})\right)}_{\rm{term\ 1}}\\
&\quad+\!\underbrace{\left(L_{\rho}(\{\mathbf{x}^{k+1}_q\}_{q=1}^Q, \mathbf{x}_0^{k+1}, \mathbf{y}^{k})-L_{\rho}(\{\mathbf{x}^{k}_q\}_{q=1}^Q, \mathbf{x}_0^{k}, \mathbf{y}^{k})\right)}_{\rm{term\ 2}}. \nonumber
\end{split}
\end{equation}
For} the first term, we have the following derivations
\begin{align}\label{eq:extra}
\begin{split}
&L_{\rho}(\{\mathbf{x}_q^{k+1}\}_{q=1}^{Q}, \mathbf{x}_0^{k+1}, \mathbf{y}^{k+1})-L_{\rho}(\{\mathbf{x}_q^{k+1}\}_{q=1}^{Q}, \mathbf{x}_0^{k+1}, \mathbf{y}^{k})\\
&\!=\!{\rm Re}\big\langle \mathbf{x}_0^{k+1}\!\!-\!\!\sum_{q=1}^{Q}2^{q-1} \mathbf{x}_q^{k\!+\!1}\!, \!\mathbf{y}^{k\!+\!1}\big\rangle\!\!
-\!\!{\rm Re}\big\langle \mathbf{x}_0^{k\!+\!1}\!\!-\!\!\sum_{q=1}^{Q}2^{q\!-\!1} \mathbf{x}_q^{k\!+\!1}, \!\mathbf{y}^{k}\big\rangle \\
&={\rm Re}\big\langle \mathbf{x}_0^{k+1}-\sum_{q=1}^{Q}2^{q-1} \mathbf{x}_q^{k+1},\; \mathbf{y}^{k+1}-\mathbf{y}^{k}\big\rangle \\
&\stackrel{\rm (a)}
=\frac{1}{\rho}\|\mathbf{y}^{k+1}-\mathbf{y}^{k}\|_2^2\\
&\stackrel{\rm (b)}
\le\frac{\lambda_{\rm max}^2(\mathbf{H}^{H} \mathbf{H})}{\rho}\|\mathbf{x}_0^{k+1}-\mathbf{x}_0^{k}\|_2^2,
\end{split}
\end{align}
where ``$\stackrel{\rm (a)}{=}$'' and ``$\stackrel{\rm (b)}{\le}$'' comes from \eqref{eq:y_update} and \eqref{eq:y_difference1} respectively.\\
For the second term, we have the following derivations
\begin{align}\label{eq:extra1}
\begin{split}
&L_{\rho}(\{\mathbf{x}_q^{k+1}\}_{q=1}^{Q}, \mathbf{x}_0^{k+1}, \mathbf{y}^{k})-L_{\rho}(\{\mathbf{x}_q^{k}\}_{q=1}^{Q}, \mathbf{x}_0^{k}, \mathbf{y}^{k})\\
=&L_{\rho}(\{\mathbf{x}_q^{k+1}\}_{q=1}^{Q}, \mathbf{x}_0^{k}, \mathbf{y}^{k})-L_{\rho}(\{\mathbf{x}_q^{k}\}_{q=1}^{Q}, \mathbf{x}_0^{k}, \mathbf{y}^{k})\\
\hspace{0.5cm} +&L_{\rho}(\{\mathbf{x}_q^{k+1}\}_{q=1}^{Q}, \mathbf{x}_0^{k+1}, \mathbf{y}^{k})-L_{\rho}(\{\mathbf{x}_q^{k+1}\}_{q=1}^{Q}, \mathbf{x}_0^{k}, \mathbf{y}^{k})\\
\le& \sum_{q=1}^{Q}\Bigl({\rm Re}\left\langle\nabla_{\mathbf{x}_q}  L_{\rho}(\!\mathbf{x}_1^{k+1}\!\!\!\!,\!\cdots\!,\mathbf{x}_{q-1}^{k+1}\!,\!\mathbf{x}_q^{k+1},\! \mathbf{x}_{q+1}^k\!,\!\cdots\!,\!\mathbf{x}_Q^k,\!\mathbf{x}_0^{k},\! \mathbf{y}^{k}), \right.\\
&\hspace{1cm}\left.\mathbf{x}_q^{k+1}-\mathbf{x}_q^{k}\right\rangle-\frac{4^{q-1}\rho-\alpha_q}{2}\|\mathbf{x}_q^{k+1}-\mathbf{x}_{q}^{k}\|_2^2\Big)\\
&\hspace{0.2cm}+{\rm Re}\left\langle\nabla_{\mathbf{x}_0}  L_{\rho}(\{\mathbf{x}_q^{k+1}\}_{q=1}^{Q}, \mathbf{x}_0^{k+1}, \mathbf{y}^{k}), \mathbf{x}_0^{k+1}-\mathbf{x}_0^{k}\right\rangle\\
&\hspace{1.5cm}-\frac{\rho + \lambda_{\rm min}(\mathbf{H}^{H} \mathbf{H})}{2}\|\mathbf{x}_0^{k+1}-\mathbf{x}^{k}_0\|_2^2 \\
 \le& \sum_{q=1}^{Q}\!\!\frac{\alpha_q\!\!-\!\!4^{q-1}\rho}{2}\|\mathbf{x}_q^{k+1}\!\!-\!\!\mathbf{x}_{q}^{k}\|_2^2\!\!-\!\!\frac{\rho \!\!+\!\!\lambda_{\rm min}(\mathbf{H}^{H}\mathbf{H})}{2}\|\mathbf{x}_0^{k+1}\!\!-\!\!\mathbf{x}^{k}_0\|_2^2,
\end{split}
\end{align}
where the first inequality holds since the corresponding augmented Lagrangian functions are strongly convex \cite{boyd2004convex} and the second inequality holds since $\mathbf{x}_q^{k+1}$ and $\mathbf{x}_0^{k+1}$ are minimizers of the problems \eqref{eq:x_q_update} and \eqref{eq:x0_update}, i.e.,
\[
\begin{split}
&\big<\nabla_{\mathbf{x}_q} L_{\rho}(\mathbf{x}_q^{k+1},\!\mathbf{x}_1^{k+1}\!\!\!\!,\!\cdots\!,\mathbf{x}_{q-1}^{k+1}\!, \mathbf{x}_{q+1}^k\!,\!\cdots\!,\!\mathbf{x}_Q^k,\!\mathbf{x}_0^{k},\! \mathbf{y}^{k}),\\
&\hspace{6cm}\mathbf{x}_q^{k}\!-\!\mathbf{x}_q^{k+1}\big> \ge 0 \\
&\nabla_{\mathbf{x}_0}  L_{\rho}(\{\mathbf{x}_q^{k+1}\}_{q=1}^{Q}, \mathbf{x}_0^{k+1}, \mathbf{y}^{k})= 0.
\end{split}
\]
Adding both sides of inequalities \eqref{eq:extra} and \eqref{eq:extra1} and letting $\gamma_q(\rho)=4^{q-1}\rho-\alpha_q$ and $\gamma_0(\rho)=\rho + \lambda_{\rm min}(\mathbf{H}^{H} \mathbf{H})$, we can obtain
\begin{align*}
&L_{\rho}(\{\mathbf{x}_q^{k+1}\}_{q=1}^{Q}, \mathbf{x}_0^{k+1}, \mathbf{y}^{k+1})-L_{\rho}(\{\mathbf{x}_q^{k}\}_{q=1}^{Q}, \mathbf{x}_0^{k}, \mathbf{y}^{k})\\
&\le \sum_{q=1}^{Q}-\frac{\gamma_q(\rho)}{2} \|\mathbf{x}_q^{k+1}-\mathbf{x}_{q}^{k}\|_2^2\\
&\quad\  -\Big(\frac{\gamma_0(\rho)}{2}-\frac{\lambda_{\rm max}^2(\mathbf{H}^{H} \mathbf{H})}{\rho}\Big) \|\mathbf{x}_0^{k+1}-\mathbf{x}_0^{k}\|_2^2,
\end{align*}
which completes the proof.
$\hfill\blacksquare$
\end{proof}

\begin{lemma}\label{lemma:L_bounded1}
Let $\alpha_q$, $\forall q\in\{1,\dotsb,Q\}$, and $\rho$ satisfy $4^{q-1}\rho>\alpha_q$ and $\rho > \sqrt{2} \lambda_{\rm max}(\mathbf{H}^H\mathbf{H})$. Assume tuples $\{\{\mathbf{x}_q^{k}\}_{q=1}^{Q},\mathbf{x}_0^k,\mathbf{y}^k\}$ is generated by Algorithm 1, then $L_{\rho}(\{\mathbf{x}_q^{k}\}_{q=1}^{Q},\mathbf{x}_0^k,\mathbf{y}^k)$ is lower bounded as follows
\begin{align}\label{eq:Lbound}
\begin{split}
&L_{\rho}(\{\mathbf{x}_q^{k}\}_{q=1}^{Q},\mathbf{x}_0^k,\mathbf{y}^k) \ge \ell\big(\sum_{q=1}^{Q}2^{q-1} \mathbf{x}_q^{k}\big)-\sum_{q=1}^{Q}\frac{\alpha_q}{2}\Vert\mathbf{x}_q^{k} \Vert_2 ^{2}.
\end{split}
\end{align}
\end{lemma}
\begin{proof}
Plugging \eqref{eq:z_nabla_expression} into \eqref{eq:lagrangian_PSADMM}, we obtain
\begin{align}\label{eq:Lbound1}
\begin{split}
&L_{\rho}(\{\mathbf{x}_q^{k}\}_{q=1}^{Q}, \mathbf{x}_0^{k}, \mathbf{y}^{k})=\ell\left(\mathbf{x}_0^{k}\right)-\sum_{q=1}^{Q}\frac{\alpha_q}{2}\Vert\mathbf{x}_q^{k} \Vert_2 ^{2}\\
&\!+\!{\rm Re}\Big\langle \!\sum_{q=1}^{Q}2^{q-1} \mathbf{x}_q^{k}\!-\!\mathbf{x}_0^{k},\!\nabla_{\mathbf{x}_0} \ell\left(\mathbf{x}_0^{k}\right)\!\Big\rangle \!+\!\frac{\rho}{2}\Big\|\mathbf{x}_0^{k}\!-\!\sum_{q=1}^{Q}2^{q-1}\mathbf{x}_q^{k}\Big\|_2^2.\\
\end{split}
\end{align}
Since we show that gradient $\|\nabla_{\mathbf{x}_0}\ell\left(\mathbf{x}_0\right)\|_2$ is Lipschitz continuous in Lemma \ref{lemma:z1} and $ \|\nabla_{\mathbf{x}_0}^2 \ell\left(\mathbf{x}_0\right)\|_2 \le \lambda_{\rm max}(\mathbf{H}^H\mathbf{H})$, according to the Decent Lemma \cite{bertsekas1999nonlinear}, we can obtain
\begin{align*}
\ell\big(\sum_{q=1}^{Q}2^{q-1} \mathbf{x}_q^{k}\big) &\!\le\! \ell\left(\mathbf{x}_0^{k}\right)\!+\!{\rm Re}\Big\langle \!\nabla_{\mathbf{x}_0} \ell\left(\mathbf{x}_0^{k}\right)\!,\!\sum_{q=1}^{Q}2^{q-1} \mathbf{x}_q^{k}-\mathbf{x}_0^{k} \!\Big\rangle \\
&\quad+\frac{\lambda_{\rm max}(\mathbf{H}^H\mathbf{H})}{2}\Big\|\sum_{q=1}^{Q}2^{q-1}\mathbf{x}_q^{k}-\mathbf{x}_0^{k}\Big\|_2^2,
\end{align*}
which can be further derived to the following inequality
\begin{align}\label{eq:Lbound2}
\begin{split}
&\ell\left(\mathbf{x}_0^{k}\right)+{\rm Re}\Big\langle \sum_{q=1}^{Q}2^{q-1} \mathbf{x}_q^{k}-\mathbf{x}_0^{k},\nabla_{\mathbf{x}_0} \ell\left(\mathbf{x}_0^{k}\right) \Big\rangle \\
&\ge \ell\big(\sum_{q=1}^{Q}2^{q-1} \mathbf{x}_q^{k}\big)-\frac{\lambda_{\rm max}(\mathbf{H}^H\mathbf{H})}{2}\Big\|\sum_{q=1}^{Q}2^{q-1}\mathbf{x}_q^{k}-\mathbf{x}_0^{k}\Big\|_2^2.
\end{split}
\end{align}
Plugging \eqref{eq:Lbound2} into \eqref{eq:Lbound1}, we can get
\begin{align}\label{eq:Lbound3}
\begin{split}
&L_{\rho}(\{\mathbf{x}_q^{k}\}_{q=1}^{Q},\mathbf{x}_0^k,\mathbf{y}^k) \ge \ell\big(\sum_{q=1}^{Q}2^{q-1} \mathbf{x}_q^{k}\big)-\sum_{q=1}^{Q}\frac{\alpha_q}{2}\Vert\mathbf{x}_q^{k} \Vert_2 ^{2}\\
&\hspace{2cm} +\frac{\rho-\lambda_{\rm max}(\mathbf{H}^H\mathbf{H})}{2}\Big\|\mathbf{x}_0^{k}-\sum_{q=1}^{Q}2^{q-1}\mathbf{x}_q^{k}\Big\|_2^2.
\end{split}
\end{align}
Since $\ell\big(\sum_{q=1}^{Q}2^{q-1} \mathbf{x}_q^{k}\big)-\sum_{q=1}^{Q}\frac{\alpha_q}{2}\Vert\mathbf{x}_q^{k} \Vert_2 ^{2}$ is bounded over $\mathbf{x}_{qR}, \mathbf{x}_{qI} \in [-1\ 1]^U$, as well as the fact that $\rho-\lambda_{\rm max}(\mathbf{H}^H\mathbf{H})>0$ comes from $\rho> \sqrt{2} \lambda_{\rm max}(\mathbf{H}^H\mathbf{H})$. Using these two cases leads to the desired result that $L_{\rho}(\{\mathbf{x}_q^{k}\}_{q=1}^{Q},\mathbf{x}_0^k,\mathbf{y}^k)$ is lower bounded and Lemma \ref{lemma:L_bounded1} has been proved.

$\hfill\blacksquare$
\end{proof}

\section{Proof of Theorem 1}\label{PS-ADMM Proof}
According to Lemma 2, summing both sides of the inequality \eqref{eq:L_difference} when $k=1,2,\cdots,+\infty$, we can obtain
\begin{equation}
    \begin{split}
    &L_{\rho}\left (\{\mathbf{x}_q^{1}\}_{q=1}^{Q}, \mathbf{x}_0^{1}, \mathbf{y}^{1} \right)-\lim_{k\rightarrow +\infty}L_{\rho}\left (\{\mathbf{x}_q^{k}\}_{q=1}^{Q}, \mathbf{x}_0^{k}, \mathbf{y}^{k}\right)\\
      &\geq \sum_{k=1}^{+\infty}\sum_{q=1}^{Q}\frac{\gamma_q(\rho)}{2}\|\mathbf{x}_q^{k+1}-\mathbf{x}_{q}^{k}\|_2^2 \\
&+\sum_{k=1}^{+\infty}\Big(\frac{\gamma_0(\rho)}{2}-\frac{\lambda_{\rm max}^2(\mathbf{H}^{H} \mathbf{H})}{\rho}\Big)\|\mathbf{x}_0^{k+1}-\mathbf{x}_0^k\|_2^2.\\
  \end{split}
  \end{equation}
From Lemma 3, one can see that $\lim_{k\rightarrow +\infty}L_{\rho}\left (\{\mathbf{x}_q^{k}\}_{q=1}^{Q}, \mathbf{x}_0^{k}, \mathbf{y}^{k}\right)>-\infty$. Moreover, since $\frac{\gamma_0(\rho)}{2}-\frac{\lambda_{\rm max}^2(\mathbf{H}^{H} \mathbf{H})}{\rho}\geq0$, we can obtain
\begin{subequations}
\begin{align}
 &\lim\limits_{k\rightarrow+\infty} \|\mathbf{x}_0^{k+1}-\mathbf{x}_0^k\|_2 = 0, \label{eq:differenceX0} \\
 &\lim\limits_{k\rightarrow+\infty}\|\mathbf{x}_q^{k+1}-\mathbf{x}_{q}^{k}\|_2 = 0,\ \  q\in\{1,\dotsb,Q\}. \label{eq:differenceXq}
 \end{align}
\end{subequations}
Plugging \eqref{eq:differenceX0} into RHS of equation \eqref{eq:y_difference1}, we get
\begin{equation}\label{limYzero}
\lim\limits_{k\rightarrow+\infty}\|\mathbf{y}^{k+1}-\mathbf{y}^k\|_2 = 0.
\end {equation}
Plugging \eqref{limYzero} into \eqref{eq:y_update}, we get
\begin{align}\label{eq:differenceX01}
\lim\limits_{k\rightarrow+\infty}\|\mathbf{x}_0^{k+1}-\sum_{q=1}^{Q}2^{q-1}\mathbf{x}_q^{k+1} \|_2 = 0.
\end{align}
Since $\mathbf{x}_{q \rm R}, \mathbf{x}_{q \rm I} \in [-1\ 1]^U$, we can obtain the following convergence results from \eqref{eq:differenceXq}.
\begin{equation}\label{convergence xq}
\lim\limits_{k\rightarrow+\infty}\mathbf{x}_q^{k}\!=\!\mathbf{x}_q^{*},~\forall~q=1,2,\cdots,Q.
\end{equation}
Plugging \eqref{convergence xq} into \eqref{eq:differenceX01}, we can see
\begin{equation}\label{convergence_x0}
\lim\limits_{k\rightarrow+\infty}\mathbf{x}_0^{k}=\mathbf{x}_0^{*}=\sum_{q=1}^{Q}2^{q-1}\mathbf{x}_q^{*}.
 \end{equation}
From \eqref{eq:z_nabla_expression},  we can derive
  \begin{equation}\label{limy_x0}
    \lim\limits_{k\rightarrow+\infty} \mathbf{y}^k = \lim\limits_{k\rightarrow+\infty}-\nabla_{\mathbf{x}_0} \ell\left(\mathbf{x}_0^{k}\right).
   \end{equation}
Since $\|\nabla_{\mathbf{x}_0} \ell\left(\mathbf{x}_0^{k+1}\right)-\nabla_{\mathbf{x}_0} \ell\left(\mathbf{x}_0^{k}\right)\|_2^2 \le \lambda_{\rm max}^2(\mathbf{H}^{H} \mathbf{H}) \|\mathbf{x}_0^{k+1}-\mathbf{x}_0^{k}\|_2^2$ and $\mathbf{x}^{k}_0$ is bounded, we can conclude that all the elements in $\nabla_{\mathbf{x}_0} \ell\left(\mathbf{x}_0\right)$ are also bounded. Therefore, equation \eqref{limYzero} indicates
\begin{equation}\label{convergence z}
\lim\limits_{k\rightarrow+\infty} \mathbf{y}^k =  \mathbf{y}^{*}.
\end{equation}

Then, we prove $\{\mathbf{x}_q^*\}_{q=1}^{Q}$ is a stationary point of problem \eqref{eq:PS_ML}.

Since $\{\mathbf{x}_q^{k+1}\}_{q=1}^{Q} = \underset{\mathbf{x}_q \in \tilde{\mathcal{X}}_q^{U}}{\arg\min} \; L_{\rho}\left (\{\mathbf{x}_q\}_{q=1}^{Q}, \mathbf{x}^{k}_0, \mathbf{y}^{k} \right) $ in \eqref{eq:x_q_update}, and $ L_{\rho}\left (\{\mathbf{x}_q\}_{q=1}^{Q}, \mathbf{x}^{k}_0, \mathbf{y}^{k} \right)$ is strongly convex w.r.t.\ $\mathbf{x}_q$, we have the following optimality conditions.

\begin{equation}\label{optimality_xq}
\begin{split}
&{\rm Re}\Big\langle\nabla_{\mathbf{x}_q}\Big(\ell\left(\mathbf{x}_0^{k}\right)-\sum_{q=1}^{Q}\frac{\alpha_q}{2}\Vert\mathbf{x}_q^{k+1} \Vert_2 ^{2}\\
&\!+\!\big\langle\!\mathbf{x}_0^{k}\!-\!\sum_{q=1}^{Q}2^{q\!-\!1} \mathbf{x}_q^{k\!+\!1}, \mathbf{y}^{k}\big\rangle\!+\!\frac{\rho}{2}\big\|\mathbf{x}_0^{k}\!-\!\sum_{q=1}^{Q}2^{q\!-\!1}\mathbf{x}_q^{k\!+\!1} \big\|_2^2\!\Big),\\
&\hspace{0.1cm} \mathbf{x}_q-\mathbf{x}_q^{k+1} \Big\rangle \ge 0,  \quad \forall~\mathbf{x}_q\in \tilde{\mathcal{X}}_q^{U}, q=1,2,\cdots,Q.
\end{split}
\end{equation}
When $k\rightarrow+\infty$, plugging convergence result \eqref{convergence_x0} into \eqref{optimality_xq}, it can be simplified as
\begin{align*}
 &{\rm Re}\Big \langle\nabla_{\mathbf{x}_q} \Big(\ell\big(\sum_{q=1}^{Q}2^{q-1} \mathbf{x}_q^*\big)-\sum_{q=1}^{Q}\frac{\alpha_q}{2}\Vert\mathbf{x}_q^* \Vert_2 ^{2}\Big),\mathbf{x}_q-\mathbf{x}^*_q\Big\rangle
 \ge 0,\\
 &\hspace{2cm}\forall~\mathbf{x}_q\in \tilde{\mathcal{X}}_q^{U},\; q=1,2,\cdots,Q.
\end{align*}
which completes the proof.$\hfill\blacksquare$

\section{Proof of Theorem 2}\label{Iteration complexity Proof}
To be clear, here we rewrite \eqref{eq:L_difference} as
\[
  \begin{split}
    &L_{\rho}\left (\{\mathbf{x}_q^{k}\}_{q=1}^{Q}, \mathbf{x}_0^{k}, \mathbf{y}^{k} \right)-L_{\rho}\left (\{\mathbf{x}_q^{k+1}\}_{q=1}^{Q}, \mathbf{x}_0^{k+1}, \mathbf{y}^{k+1}\right)\\
      &\geq \sum_{q=1}^{Q}-\frac{\gamma_q(\rho)}{2}\|\mathbf{x}_q^{k+1}-\mathbf{x}_{q}^{k}\|_2^2 \nonumber\\
&-\Big(\frac{\gamma(\rho)}{2}-\frac{\lambda_{\rm max}^2(\mathbf{H}^{H} \mathbf{H})}{\rho}\Big)\|\mathbf{x}_0^{k+1}-\mathbf{x}_0^k\|_2^2.\\
  \end{split}
\]

According to Lemma \ref{lemma:L_difference}, there exists a constant $C\!=\!\min\!\left\{\!\{ \frac{\gamma_q(\rho)}{2}\}_{q=1}^{Q}, \Big(\frac{\gamma(\rho)}{2}\!\!-\!\!\frac{\lambda_{\rm max}^2(\mathbf{H}^{H} \mathbf{H})}{\rho}\Big)\! \right\}$
such that
\[
  \begin{split}
    &L_{\rho}\left (\{\mathbf{x}_q^{k}\}_{q=1}^{Q}, \mathbf{x}_0^{k}, \mathbf{y}^{k} \right)-L_{\rho}\left (\{\mathbf{x}_q^{k+1}\}_{q=1}^{Q}, \mathbf{x}_0^{k+1}, \mathbf{y}^{k+1}\right)\\
      &\geq C\Big(\sum_{q=1}^{Q}\|\mathbf{x}_q^{k+1}-\mathbf{x}_{q}^{k}\|_2^2 +\|\mathbf{x}_0^{k+1}-\mathbf{x}_0^k\|_2^2\Big).\\
  \end{split}
\]
 Summing both sides of the above inequality from $k=1,\cdots, K$, we have
 \begin{equation}\label{dL}
   \begin{split}
      &L_{\rho}\left (\{\mathbf{x}_q^{1}\}_{q=1}^{Q}, \mathbf{x}_0^{1}, \mathbf{y}^{1} \right)-L_{\rho}\left (\{\mathbf{x}_q^{K+1}\}_{q=1}^{Q}, \mathbf{x}_0^{K+1}, \mathbf{y}^{K+1}\right)\\
      &\geq \sum_{k=1}^{K}\Bigg(C\Big(\sum_{q=1}^{Q}\|\mathbf{x}_q^{k+1}-\mathbf{x}_{q}^{k}\|_2^2 +\|\mathbf{x}_0^{k+1}-\mathbf{x}_0^k\|_2^2\Big)\Bigg).\\
  \end{split}
 \end{equation}
 Since $t = \underset{k}{\rm min}\{k|\sum_{q=1}^{Q}\|\mathbf{x}_q^{k+1}-\mathbf{x}_{q}^{k}\|_2^2 + \|\mathbf{x}_0^{k+1}-\mathbf{x}_0^{k}\|_2^2\leq\epsilon\}$, we can change \eqref{dL} to
 \begin{equation}\label{Relax L}
   \begin{split}
    & L_{\rho}\left (\{\mathbf{x}_q^{1}\}_{q=1}^{Q}, \mathbf{x}_0^{1}, \mathbf{y}^{1} \right)-L_{\rho}\left (\{\mathbf{x}_q^{K+1}\}_{q=1}^{Q}, \mathbf{x}_0^{K+1}, \mathbf{y}^{K+1}\right)\\
    &  \geq tC\epsilon.
   \end{split}
 \end{equation}
 Since we have $L_{\rho} \!\left (\{\mathbf{x}_q^{K+1}\}_{q=1}^{Q}, \mathbf{x}_0^{K+1}, \mathbf{y}^{K+1}\right)\! \geq\! L_{\rho}\! \left (\{\mathbf{x}_q^*\}_{q=1}^{Q}, \mathbf{x}_0^*, \mathbf{y}^*\right)$, \eqref{Relax L} can be reduced to
 \[
   \begin{split}
     t \!&\leq\! \frac{1}{C\epsilon}\bigg(\!L_{\rho}\left (\{\mathbf{x}_q^{1}\}_{q=1}^{Q}, \mathbf{x}_0^{1}, \mathbf{y}^{1} \right) \!-\! L_{\rho}\left (\{\mathbf{x}_q^*\}_{q=1}^{Q}, \mathbf{x}_0^*, \mathbf{y}^*\right)\!\bigg),
   \end{split}
 \]
 where $L_{\rho}\left (\{\mathbf{x}_q^*\}_{q=1}^{Q}, \mathbf{x}_0^*, \mathbf{y}^*\right)=\ell\left(\mathbf{x}_0^*\right) - \sum_{q=1}^{Q}\frac{\alpha_q}{2} \Vert\mathbf{x}_q^*\Vert_2 ^{2}$, which concludes the proof of Theorem \ref{iteration complexity}.
 $\hfill\blacksquare$

\ifCLASSOPTIONcaptionsoff
  \newpage
\fi


\end{document}